\documentclass{article}
\usepackage{amsmath, amssymb,amsthm, bm, graphicx}
\usepackage{url} 
\usepackage{tikz}
\usetikzlibrary{arrows,decorations.pathreplacing}

\usepackage{subcaption}
\usepackage{booktabs} 

\usepackage[margin=1in]{geometry}
\usepackage{natbib}        % Collapses numeric citations like [1,2,3] → [1–3]
\usepackage{enumitem}
\usepackage{hyperref}
\title{A flexible class of latent variable models for the analysis of antibody response data}

\author{
Emanuele Giorgi$^{1,2}$ \and
Jonas Wallin$^{3}$ 
}

\newtheorem{theorem}{Theorem}[section]

\newtheorem{lemma}{Lemma}[section]
\newtheorem{assumption}{Assumption}

\newcommand{\E}{\mathbb{E}}

\date{
$^{1}$Department of Applied Health Science, University of Birmingham, Birmingham, UK\\
$^{2}$ Faculty of Health and Medicine, Lancaster University, Lancaster, UK\\
$^{3}$Department of Statistics, Lund University, Lund, Sweden\\
}

\usepackage{float}            % adds [ht!]
\usepackage[section]{placeins}% enables \FloatBarrier and keeps floats inside sections

\begin{document}
\maketitle

\begin{abstract}
Existing approaches to modelling antibody concentration data are mostly based on finite mixture models that rely on the assumption that individuals can be divided into two distinct groups: seronegative and seropositive. Here, we challenge this dichotomous modelling assumption and propose a latent variable modelling framework in which the immune status of each individual is represented along a continuum of latent seroreactivity, ranging from minimal to strong immune activation. This formulation provides greater flexibility in capturing age-related changes in antibody distributions while preserving the full information content of quantitative measurements. We show that the proposed class of models can accommodate a large variety of model formulations, both mechanistic and regression-based, and also includes finite mixture models as a special case. We also propose a computationally efficient $L_2$-based estimator as an alternative to maximum likelihood estimation, which substantially reduces computational cost, and we establish its consistency. Through a case study on malaria serology, we demonstrate how the flexibility of the novel framework enables joint analyses across all ages while accounting for changes in transmission patterns. We conclude by outlining extensions of the proposed modelling framework and its relevance to other omics applications.
\end{abstract}

\section{Introduction}

Serological data analysis plays a central role in reconstructing individual- and population-level exposure histories to infectious diseases, using information derived from antibody concentration measurements \citep{metcalf2016, corran2007}. Antibody levels can function as biomarkers of prior exposure, enabling the detection of both recent infections and long-term immune responses. In low-prevalence or post-elimination settings, serological data become particularly valuable as they can be used to identify subtle signs of disease resurgence or reintroduction \citep{Drakeley2005, arnold2018}. Consequently, serological surveillance offers a powerful approach for monitoring disease dynamics, especially where traditional diagnostic methods may be insufficiently sensitive due to sporadic infections or predominantly asymptomatic transmission \citep{Drakeley2005}.

Gaussian mixture models (GMMs) are among the most widely used methods for analysing antibody concentration data \citep{fraley2002, arnold2018}. In the serological context, a GMM typically assumes that the antibody concentration $y_i$ for the $i$-th individual arises from a mixture of two Gaussian distributions, representing seronegative and seropositive subpopulations:
\begin{equation}
    \label{eq:gmm}
    f(y) = \pi_0 \mathcal{N}(y ; \mu_0, \sigma_0^2) + \pi_1 \mathcal{N}(y ; \mu_1, \sigma_1^2),
\end{equation}
where $\pi_0 + \pi_1 = 1$, and $(\mu_k,\sigma_k^2)$ denote the mean and variance of each component. This binary specification reflects the assumption that individuals fall into one of two mutually exclusive serological states, with no allowance for intermediate or uncertain exposure histories. Across a wide range of pathogens, such mixture models have become a standard way to operationalise immune activation from serological measurements by separating individuals into low- and high-response groups interpreted as immunologically naïve versus previously exposed or vaccinated. Although GMMs provide a flexible, data-driven approach to defining seropositivity thresholds, alternative strategies are also common in sero-epidemiology \citep{Hay2024Serodynamics}: using manufacturer-defined assay cutoffs; determining thresholds via ROC curves when positive and negative controls are available; basing cutoffs on the distribution of negative controls alone (typically two or three standard deviations above the mean); or, when available, applying immunological correlates of protection to define seroprotection-based thresholds.

In this paper, we challenge this dichotomous assumption by proposing a latent variable framework in which immune status is represented along a continuous scale of seroreactivity. Rather than modeling antibody concentrations as a mixture of outcome distributions, we shift the modelling building process from the outcome space to a latent space. By defining an individual's  antibody distribution conditionally on this continuous latent process, we demonstrate that this formulation provides greater flexibility and more reliably reflects age-related variation in antibody distributions, with standard finite mixture models arising as a special case of the proposed modelling framework.

This approach builds on an extensive body of work applying GMMs to diverse pathogens for which serology provides a reliable marker of past exposure, including directly transmitted childhood infections such as measles, mumps and rubella \citep{Vyse2006}; parasitic and neglected tropical diseases such as onchocerciasis \citep{Golden2016} and trachoma \citep{Migchelsen2017}; mosquito-borne infections such as malaria, where mixture models underpin estimates of age-specific seroprevalence \citep{Sepulveda2015}; and viral infections characterised by substantial asymptomatic transmission such as hepatitis~E \citep{Katuwal2024} and SARS-CoV-2 \citep{Bottomley2021}. These applications highlight that the usefulness of serology strongly depends on the biology of the pathogen. As emphasised by \citet{metcalf2016}, serological surveys are most informative for infections in which antibody responses provide a durable and interpretable record of prior exposure or immunity. This includes infections that generate long-lasting protection, those with measurable antigenic variation, and those in which antibodies reliably reflect cumulative exposure even when infections are predominantly asymptomatic or clinically under-ascertained. The methods presented in this paper are relevant for the analysis of these types of pathogens.

Recent methodological advances have sought to extend GMMs by incorporating covariates, such as age or other demographic factors, which influence antibody concentrations and seropositivity status. For example, \citet{hardelid2008} used maternal age and country-of-birth to model variation in rubella antibody levels among newborns using a finite-mixture regression framework. Age-dependent mixture models are especially common in infectious disease applications: \citet{delfava2016} embedded a three-component mixture within a catalytic model to estimate age-specific immunity to varicella-zoster virus, while \citet{cox2022} used penalised splines to let seropositivity probabilities vary smoothly with age in dengue serology data. Similarly, \citet{kyomuhangi2021} combined antibody acquisition and catalytic models to allow for temporally varying seroconversion and boosting rates. These approaches reflect cumulative exposure risk by allowing mixing proportions or transition probabilities to evolve with age or other covariates. In parallel, efforts have been made to relax the normality assumption within each component. A key extension replaces Gaussian components with skew-normal distributions \citep{azzalini1985}, enabling the model to capture asymmetries commonly observed in antibody concentration data. Skew-normal mixtures retain analytical tractability while offering greater flexibility in representing the shape of seronegative and seropositive subpopulations \citep{dias_domingues2024}.

An alternative line of work has sought to avoid explicit dichotomisation through antibody acquisition (AA) models, in which age-dependent boosting and waning processes generate smooth trajectories of mean antibody levels \citep{yman2016}. Although these models provide a useful mechanistic link between age, exposure, and antibody dynamics, they typically rely on a log-Gaussian assumption and do not adequately accommodate the skewness and multimodality that characterise many serological datasets. Semiparametric methods have also been proposed that use ensemble machine learning to estimate age-antibody curves nonparametrically \citep{arnold2017}, but these target only the mean antibody level, precluding inference on other parameters of epidemiological interest. The proposed framework in this study unifies the strengths of these approaches by accommodating both mechanistic interpretability through AA-type formulations and flexible modeling of the full heterogeneous distribution of antibody levels across age through a latent variable formulation.

A key practical challenge in implementing the proposed framework is the computational burden associated with maximum likelihood estimation, which requires repeated numerical integration. To address this, we develop a computationally efficient alternative based on histogram-based $L_2$ minimum distance estimation that substantially reduces the computational burden. Our work builds on the $L_2$ error theory developed by \citet{scott2001parametric}, which establishes the criteria under which $L_2$ estimators converge to the true parameter values. We extend this framework to our latent variable modelling framework and establish consistency of the resulting estimator under appropriate regularity conditions. 

The remainder of the paper is organised as follows. Section~\ref{sec:framework} introduces the proposed latent variable modelling framework, contrasts it with alternative latent formulations, and presents both single-density and mixture specifications for the latent seroreactivity distribution. Section~\ref{sec:age_dependency} describes how age-dependent structure can be incorporated through mechanistic and data-driven parameterisations of both single Beta distributions and finite mixtures. Section~\ref{sec:inference} details maximum likelihood and histogram-based $L_2$ estimation procedures, establishes the consistency of the $L_2$ estimator, and describes model validation approaches. Section~\ref{sec:simulation} presents two simulation studies: the first compares the statistical and computational efficiency of $L_2$ and maximum likelihood estimation, and the second evaluates the robustness of maximum likelihood estimation under model misspecification considering epidemiologically relevant data-generating mechanisms. Section~\ref{sec:applications} demonstrates the framework through applications to malaria serology data for AMA1 and MSP1 antigens, including age-stratified model comparisons. We conclude in Section~\ref{sec:discussion} with methodological implications and extensions to multivariate settings, spatial-temporal modelling, and high-dimensional omics applications.

The R scripts of the analyses presented in this paper can be freely obtained at \url{github.com/giorgistat/latentv-model-paper}. These can be run on a simulated data-set available from the same link.

\section{The proposed latent variable modelling framework for serological data analysis}
\label{sec:framework}

To motivate our modelling framework, we conceptualise immune activation in the population as a continuum ranging from immunologically quiescent states (no activation) to fully activated responses. Rather than imposing a strict seronegative–seropositive divide, we allow for individuals to occupy intermediate levels of activation that reflect gradual initiation of immune effector mechanisms and their subsequent regulation. We then adopt a latent stochastic process to capture this full continuum of immune activation.

Figure~\ref{fig:prior-models} illustrates the distribution of latent immune activation across the population. We denote by $T_i\in[0,1]$ the latent variable representing the immunological activation state of individual $i$, for $i=1,\ldots,n$. This continuous measure quantifies the underlying immune response to the antigen of interest, with values near 0 corresponding to low-activation states characterised by minimal serological activity, and values near 1 indicating high-activation states reflecting strong or boosted antibody responses.

In this framework, we use the term \emph{seroreactivity} to denote the presence or degree of antibody reactivity, irrespective of any diagnostic threshold, and describe individuals with non-zero serological activity as
\emph{seroreactive}. This concept is distinct from \emph{seropositivity}, which refers to crossing an assay-specific cut-off and therefore yields a binary classification. Seroreactivity is inherently continuous and may be
observed at levels below the positivity threshold, whereas seropositivity represents only the upper portion of this continuum.

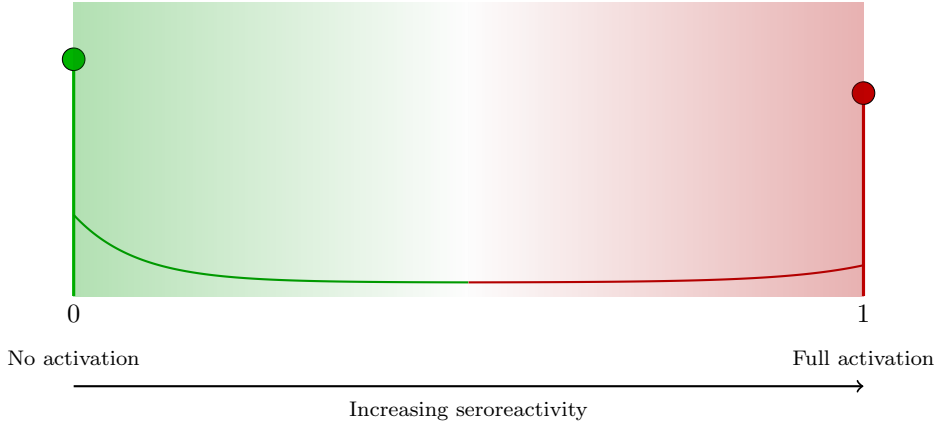
\begin{figure}[ht!]
\centering
\begin{tikzpicture}[scale=1.5]
  \node[below] at (0,0) {0};
  \node[below] at (7,0) {1};
  \node[below] at (0,-0.4) {\footnotesize No activation};
  \node[below] at (7,-0.4) {\footnotesize Full activation};

  \draw[->,thick] (0,-0.8) -- (7,-0.8)
    node[midway,below=2pt] {\footnotesize Increasing seroreactivity};

  \draw[very thick,green!70!black] (0,0) -- (0,2.1);
  \draw[fill=green!70!black] (0,2.1) circle (0.10);

  \draw[very thick,red!75!black] (7,0) -- (7,1.8);
  \draw[fill=red!75!black] (7,1.8) circle (0.10);

  \begin{scope}
    \clip (0,0) rectangle (3.5,2.6);
    \shade[left color=green!60!black, middle color=green!10, right color=gray!5, opacity=0.3]
      (0,0) rectangle (3.5,2.6);
    \draw[thick,green!60!black,domain=0:3.5,samples=120]
          plot(\x,{0.6*exp(-1.8*\x)+0.12});
  \end{scope}

  \begin{scope}
    \clip (3.5,0) rectangle (7,2.6);
    \shade[left color=gray!5, middle color=red!10, right color=red!70!black, opacity=0.3]
      (3.5,0) rectangle (7,2.6);
    \draw[thick,red!70!black,domain=3.5:7,samples=120]
          plot(\x,{0.1529*exp(-1.4*(7-\x))+0.12});
  \end{scope}

\end{tikzpicture}

\caption{Example of a hypothesized distribution of the latent immune activation in the general population. The latent variable $T \in [0,1]$ represents a continuous seroreactivity scale from no activation ($t=0$) to full activation ($t=1$), with most individuals near the extremes and fewer in intermediate activation states.}
\label{fig:prior-models}
\end{figure}

We assume that the observed measurements $y$, representing quantitative antibody concentrations or other continuous immunological biomarkers, are conditionally independent given the latent seroactivity level $T$. Conditional on $T = t$, the data follow a Gaussian distribution whose mean and variance evolve smoothly across the immune-activation continuum:
\begin{equation}
\label{eq:lbg}
Y \mid T = t \sim \mathcal{N}\!\left( (1 - t)\mu_0 + t \mu_1,\; (1 - t)\sigma_0^2 + t \sigma_1^2 \right).    
\end{equation}

In this specification, the parameters $\mu_0$ and $\mu_1$ represent the expected antibody levels at the two extremes of the latent seroactivity scale. The quantity $\mu_0$ corresponds to the mean antibody concentration for individuals with minimal or absent serological activation ($T = 0$), reflecting a baseline or resting immunological state. Conversely, $\mu_1$ corresponds to the mean antibody concentration at the upper end of the detectable seroreactivity scale ($T = 1$). This represents the assay’s saturation range, beyond which further antigenic stimulation does not produce a measurable increase in the observed antibody signal. The variance parameters $\sigma_0^2$ and $\sigma_1^2$ describe the heterogeneity of antibody responses at these two extremes. 

To illustrate the interpretation, let us consider the simplest formulation of antibody acquisition models originally proposed by \citet{yman2016}. According to this model, we define
\begin{align}
\label{eq:antibodySimple}
\E[Y;a] = f(a) = \mu_0 + (\mu_1 - \mu_0)\{1 - \exp(-r a)\},
\end{align}
where $a$ denotes age, $\mu_0$ is the mean antibody level at birth, $\mu_1$ is the asymptotic mean antibody level attained after prolonged exposure, and $r>0$ is the antibody decay rate in the underlying kinetic model. By construction, $f(0) = \mu_0$ and $\lim_{a \to \infty} f(a) = \mu_1$, so that $f(a)$ encodes the basic assumption that mean antibody levels increase with age and approach a saturation plateau at rate $r$, in the sense that the difference $\mu_1 - f(a)$ decays exponentially as age increases. Within the proposed latent-seroreactivity framework, this specification can be reformulated by expressing the mean of the latent variable $T$ as
$$
\E[T;a] = \frac{f(a) - \mu_0}{\mu_1 - \mu_0} = 1 - \exp(-r a),
$$
which increases monotonically with age. The original antibody acquisition model can then be viewed as a special case of our formulation in which the expected antibody concentration at age $a$ is written as
$$
\E[Y;a] = \mu_0 + (\mu_1 - \mu_0) \E[T;a],
$$
so that $\mu_0$ and $\mu_1$ retain their interpretation as the baseline and saturation antibody levels, while $\E[T]$ describes how the expected seroreactivity level increases with age according to the classical antibody acquisition dynamics driven by the decay rate $r$.

The distribution of $T$ should ideally accommodate both discrete and continuous components: a point mass at $T = 0$ to represent individuals that are not seroreactive, a point mass at $T = 1$ to represent individuals who have reached antibody saturation. However, in practice, reliably estimating both the discrete boundary masses and the continuous interior density is often infeasible, especially when working with noisy cross-sectional serological data. To overcome this limitation, we approximate the target distribution of $T$ using a continuous density defined on $[0,1]$ that retains enough flexibility to represent intermediate states while allowing the density to concentrate near the boundaries. In this way, probability mass that would otherwise be explicitly assigned to the points $T = 0$ and $T = 1$ is effectively absorbed into the tails of the continuous distribution, providing a practical approximation that captures both extreme and transitional immune states. In cases where the estimated distribution is highly skewed toward one of the extremes, with most of its mass concentrated near 0 or 1, the model could in principle be reformulated to include an explicit discrete component at that boundary together with a continuous density on the remaining interval; we show an example of this in the case study of Section \ref{sec:ama_analysis}. However, specifying such a mixed discrete–continuous distribution at the outset is generally not advisable, unless there is empirical justification for this modelling choice.

In Section \ref{sec:single_distr} and Section \ref{sec:mix_distr}, we propose two alternative strategies for modelling the latent variable $T$, each aligned with different inferential objectives. The first strategy is a single-density specification in which $T$ is modelled by a continuous distribution on $(0,1)$. This formulation is suitable when the goal is to describe the distribution of antibody concentrations in a flexible and data-driven manner without imposing strong structural assumptions. By allowing $T$ to vary continuously along the seroactivity scale, this approach naturally captures a wide range of antibody profiles and provides a more adaptable description of the data than the standard GMM in \eqref{eq:gmm}, making it particularly useful in exploratory analyses.

The second strategy models $T$ as a finite mixture of components that correspond directly to distinct seroreactivity levels. For example, one component may represent individuals with negligible seroreactivity, another may represent individuals showing high seroreactivity, and a third may capture intermediate levels reflecting early infection, recent seroconversion, waning immunity, or partial boosting.

\subsection{Comparison with alternative latent model formulations}

It is important to clarify how the conditional specification in \eqref{eq:lbg} differs from alternative latent variable formulations that might appear similar but impose different structures on the conditional distribution of $Y$ given $T$ and have fundamentally different interpretations.

One alternative is to model $Y$ as a weighted average of two independent latent outcomes:
\begin{equation}
\label{eq:stoch_mix}
Y \mid T = t \;=\; (1-t) \cdot Y_0 + t \cdot Y_1,
\end{equation}
with $Y_0 \sim \mathcal{N}(\mu_0, \sigma_0^2)$ and $Y_1 \sim \mathcal{N}(\mu_1, \sigma_1^2)$ independent of each other and of $T$. In this formulation, each individual's observed antibody level is conceptualised as a weighted average of two latent outcomes, representing low and high antibody levels, with $T$ determining the relative contribution of each. While both \eqref{eq:stoch_mix} and our specification in \eqref{eq:lbg} yield the same conditional mean $\mathbb{E}(Y \mid T=t) = (1-t)\mu_0 + t\mu_1$, they differ in their variance structure. In the stochastic mixture \eqref{eq:stoch_mix}, the conditional variance is ${\rm Var}(Y \mid T=t) = (1-t)^2 \sigma_0^2 + t^2 \sigma_1^2$, a quadratic function that attains its minimum at the interior point $t^* = \sigma_0^2/(\sigma_0^2 + \sigma_1^2)$, implying that individuals at intermediate values of $T$ exhibit lower variability than those at the extremes. In contrast, our formulation in \eqref{eq:lbg} specifies that both the mean and variance interpolate linearly with $T$. In typical serological settings, we expect $\sigma_1^2 < \sigma_0^2$ because antibody measurements near assay saturation exhibit reduced variability due to ceiling effects. Under this plausible ordering, linear interpolation naturally produces monotonically decreasing variance as $T$ increases from 0 to 1, whereas the stochastic mixture creates an unnatural U-shaped variance profile with minimal variability at intermediate states, which lacks biological justification for immune response data. Moreover, the  introduction of $Y_0$ and $Y_1$ as distinct outcome mechanisms retains the binary classification, whilst in our framework we seek to avoid such discrete serological categories.

Another alternative formulation introduces an additional layer of stochasticity by modelling $Y$ as
\begin{equation}
\label{eq:binary_mix}
Y \mid T = t \;=\; (1-Z_t) \cdot Y_0 + Z_t \cdot Y_1,
\end{equation}
where $Z_t \sim \text{Bernoulli}(t)$ is a latent binary indicator independent of $Y_0 \sim \mathcal{N}(\mu_0, \sigma_0^2)$ and $Y_1 \sim \mathcal{N}(\mu_1, \sigma_1^2)$. In this specification, $T$ governs the probability that an individual belongs to the high-response class, and the observed antibody level is drawn from one of two fixed Gaussian distributions depending on the realised value of $Z_t$. The conditional variance in this case is
$$
{\rm Var}(Y \mid T=t) = (1-t)\sigma_0^2 + t\sigma_1^2 + t(1-t)(\mu_1-\mu_0)^2,
$$
which includes an additional between-component variance term $t(1-t)(\mu_1-\mu_0)^2$ that is maximised at $t=0.5$. This formulation corresponds to a conventional GMM conditional on $t$, and remains fundamentally a binary classification framework in which individuals are assigned probabilistically to one of two discrete serological states. Although this extension might offer greater distributional flexibility than the standard GMM, it preserves the core limitation of partitioning the population into discrete serological states rather than representing immune activation as a continuous biological process. Our goal in this paper is to move away from this dichotomous paradigm entirely. By treating $T$ as a continuous measure of immune activation with smoothly varying distributional parameters, our framework \eqref{eq:lbg} represents immune responses along a continuous biological spectrum, which, as we show in this paper, yields not only improved model fits but also enhances interpretability.

\subsection{Modelling the latent state using a single-density parametric model}
\label{sec:single_distr}

\begin{figure}[ht!]
    \centering
    \includegraphics[width=0.8\linewidth]{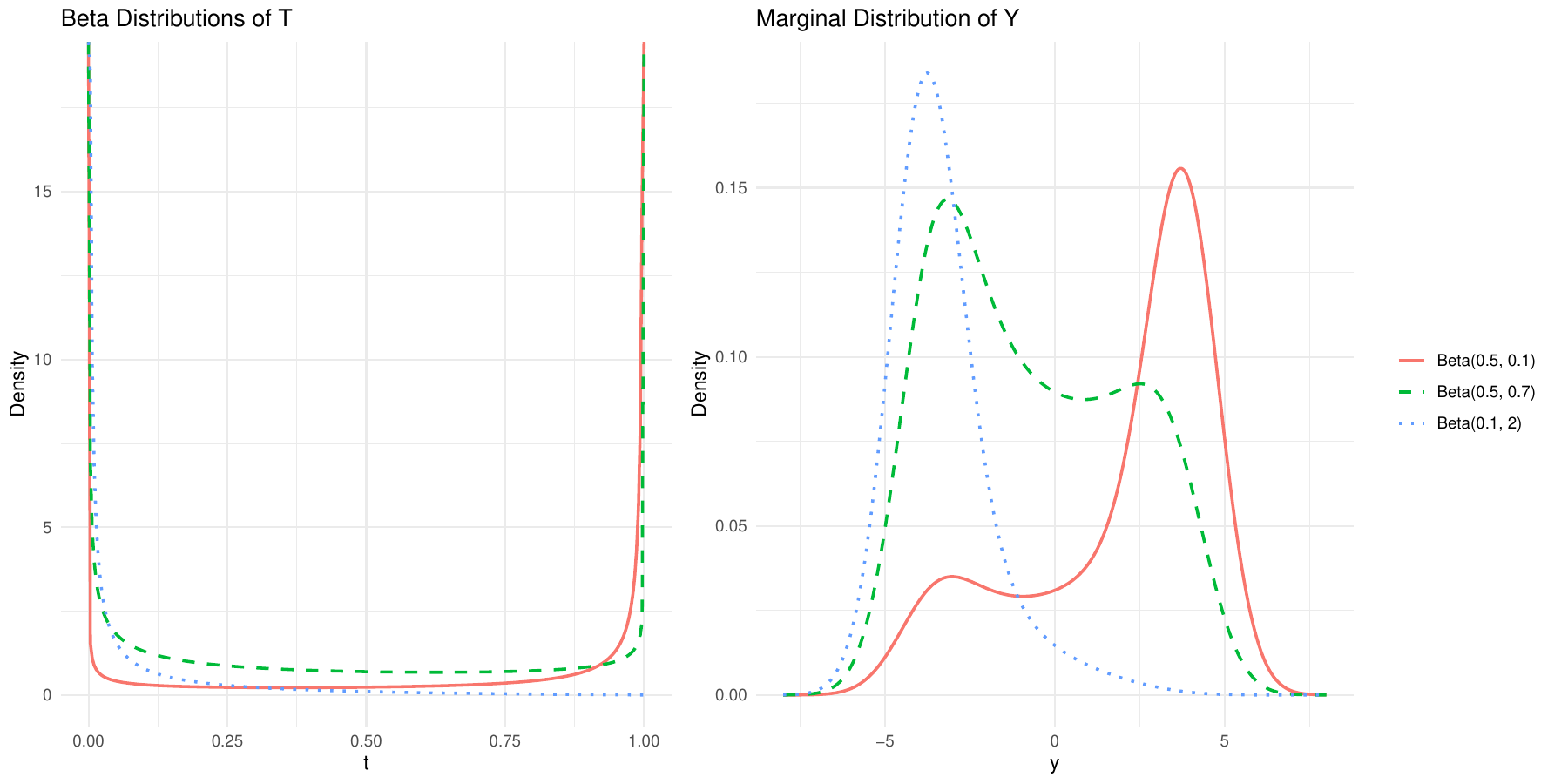}
    \caption{Example of Beta distributions for $T$ and resulting distribution for $Y$ based on the model in \eqref{eq:lbg}. Here, $\mu_0=-4$, $\mu_1=4$ and $\sigma_{0}^2=\sigma_{1}^2=1$. The value of the Beta distribution parameters are defined in the legend.}
    \label{fig:examples_single_beta}
\end{figure}

In the first approach, we define a parametric model for the latent variable $T$ using a single continuous distribution over the interval $(0,1)$. Specifically, we assume $T \sim \text{Beta}(\alpha, \beta)$ with $\E[T] = \alpha/(\alpha+\beta)$, which offers a flexible yet tractable way to represent a wide range of shapes depending on the choice of $\alpha$ and $\beta$. This allows the model to capture both gradual and sharply polarized immune response profiles, while retaining a parsimonious structure that facilitates interpretation and efficient inference. 

Specifically, the parameters $\alpha > 0$ and $\beta > 0$ determine the shape of the distribution. Larger values of $\alpha$ increase concentration near $T = 1$, whereas larger values of $\beta$ shift mass toward $T = 0$. Figure~\ref{fig:examples_single_beta} illustrates how different Beta distributions for $T$ influence the marginal distribution of the outcome $Y_i$ through the latent distribution of $T$. Different combinations of $(\alpha, \beta)$ can be interpreted in terms of epidemiological scenarios, with some configurations being more plausible than others depending on the context.

\begin{itemize}
\item $\alpha, \beta > 1$: The density peaks at intermediate values of $t$, suggesting that most individuals exhibit moderate levels of seroreactivity, with relatively few individuals showing either very low or very high antibody levels. This pattern may reflect transmission settings characterised by frequent but short-lived infections, or fluctuating exposure due to seasonal or intermittent transmission dynamics. However, such a pattern may be biologically implausible for infections that tend to induce long-lasting immunity or follow more polarized exposure patterns.

\item $\alpha < 1, \beta > 1$: The distribution is concentrated near $t = 0$, indicating low prevalence and a predominance of individuals with little or no detectable immune response. This scenario may arise in populations where transmission has been substantially reduced, for instance due to effective control interventions or ecological changes, resulting in limited exposure and predominantly seronegative individuals.

\item $\alpha > 1, \beta < 1$: The distribution places most probability mass near $t = 1$, indicating high prevalence and widespread exposure. This shape is expected in settings with intense and ongoing transmission, where individuals are frequently or continuously exposed, leading to strong and sustained immune responses across the population.

\item $\alpha, \beta < 1$: The distribution is U-shaped, with probability mass concentrated near both $t = 0$ and $t = 1$, suggesting a bimodal population. Such a distribution may emerge in contexts where the population is split between individuals with no prior exposure and those with repeated or recent exposure. This structure can also arise in settings with strongly seasonal or clustered transmission patterns.
\end{itemize}

While a single Beta distribution provides sufficient flexibility to tailor the distribution of $T$ to most epidemiological settings, it is not the only possible choice. Alternative parametric forms such as truncated Gaussian, exponential, or other bounded distributions can be used to obtain density shapes that are not captured by those described above.

\subsection{Modelling the latent state using mixture distributions}
\label{sec:mix_distr}

\begin{figure}[ht!]
    \centering
    \includegraphics[width=0.9\textwidth]{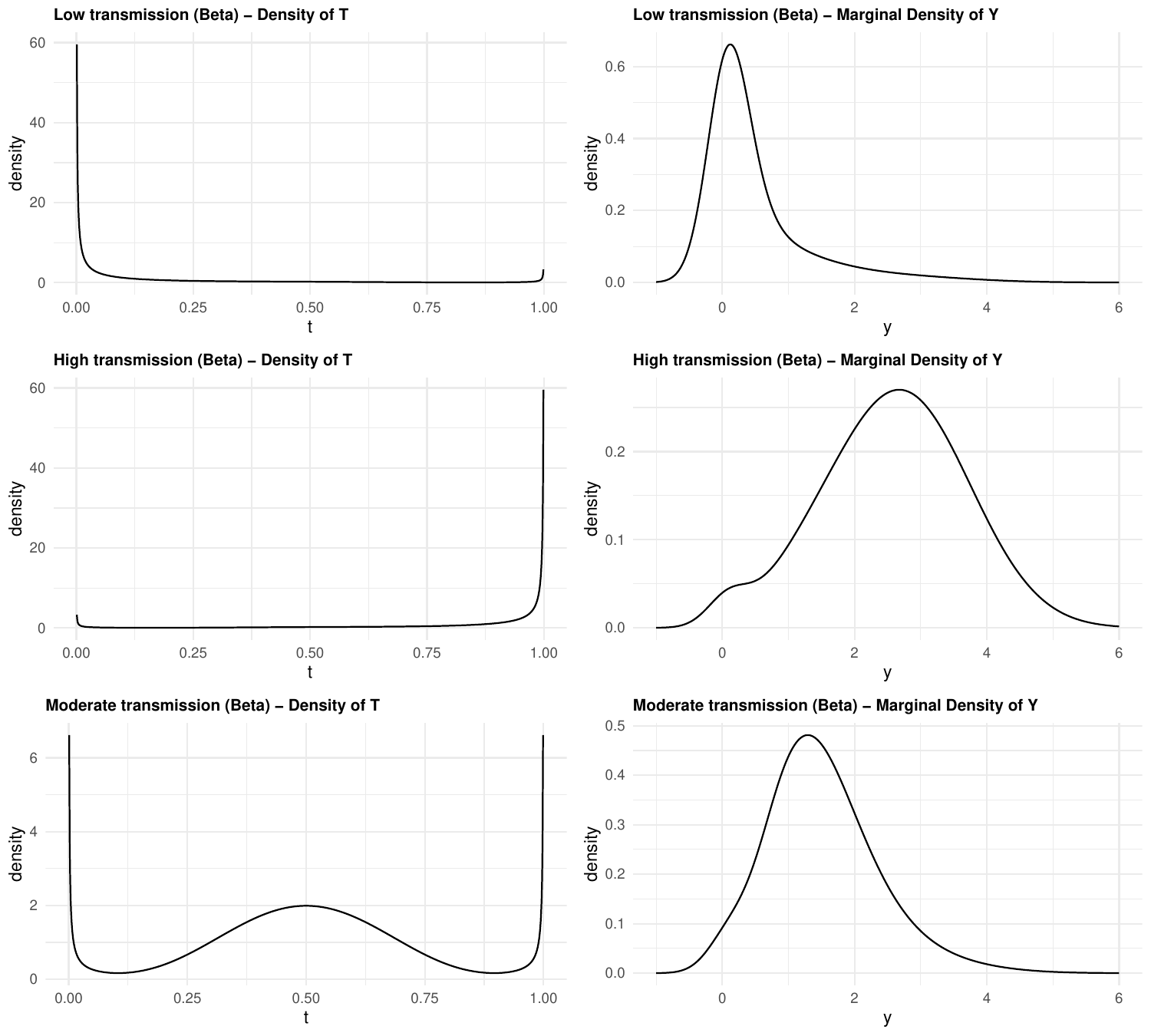}
    \caption{
    Examples of latent-state mixture models using Beta distributions. Each row represents a distinct epidemiological scenario characterised by different levels of disease transmission intensity: low transmission (top panels, predominantly seronegative individuals), high transmission (middle panels, predominantly seropositive individuals), and moderate transmission (bottom panels, predominantly intermediate immune states). The left-hand panels show the density of the latent immune response state $T$, while the right-hand panels depict the induced marginal distributions of the outcome variable $Y$.
    }
    \label{fig:beta_mixtures}
\end{figure}

An alternative to modelling the latent variable $T$ with a single parametric distribution is to assume that the population comprises distinct immunological subgroups, each associated with a different pattern of antibody response. This motivates the use of finite mixture models, in which the overall distribution of $T$ is expressed as a weighted combination of several component densities defined on the interval $(0,1)$. 

We first consider a three component mixture model to illustrate this approach. The first component captures individuals with very low seroreactivity, characterised by minimal detectable immune activation. The second component corresponds to individuals with high seroreactivity, who display strong and sustained antibody responses, typically reflecting substantial or repeated immunological stimulation. The third component represents intermediate levels of seroreactivity, for example individuals with evidence of past antigenic exposure who retain partial immune activation but exhibit moderate antibody concentrations due to waning or the absence of recent boosting. The distribution of $T$ is then given by
\begin{equation}
\label{eq:mixture_general}
f_T(t) = \sum_{c=1}^{3} \pi_c f_c(t_i; \psi_c), \quad t \in (0,1),
\end{equation}
where each $f_c(\cdot; \psi_c)$ is a component density defined on $(0,1)$ with parameters $\psi_c$, and $\pi_c$ is the mixing weight for component $c$, subject to $\sum_{c=1}^3 \pi_c = 1$ and $\pi_c \geq 0$.

For interpretability, one may specify the mixture components as Beta distributions with biologically motivated parameter constraints. For example, a low seroreactivity group can be represented with $\alpha < 1$ and $\beta > 1$, which produces a density concentrated near $t=0$. A high seroreactivity group can be represented with $\alpha > 1$ and $\beta < 1$, yielding a density concentrated near $t=1$. An intermediate seroreactivity group can be represented with $\alpha > 1$ and $\beta > 1$, producing a unimodal distribution centred in the interior of $(0,1)$. Alternatively, truncated Gaussian distributions on $(0,1)$ can be used, with modes positioned near $0$ (low seroreactivity), near $1$ (high seroreactivity), and between $0$ and $1$ (intermediate).

Figure~\ref{fig:beta_mixtures} illustrate the use and interpretation of the three-component mixture model for $T$ across different transmission settings. Each figure shows three contrasting scenarios corresponding to low transmission (dominated by individuals with low seroreactivity), high transmission (dominated by individuals with high seroreactivity), and moderate transmission (dominated by the intermediate group). The left panels display the densities of the latent immune state $T$, and the right panels show the resulting marginal distributions of the observed outcome $Y$, highlighting how assumptions about the latent structure propagate to the data. The shapes produced by a three component mixture model are, as expected, more flexible than those obtained with the single Beta model introduced in the previous section. For instance, the distribution of $T_i$ under a moderate transmission setting (lower panel of Figure~\ref{fig:beta_mixtures}), which shows a pronounced concentration of density around the middle of the interval, cannot be reproduced by a single Beta distribution.

Although conceptually appealing, the three component formulation can be difficult to identify in practice, particularly when the intermediate group overlaps substantially with the low or high seroreactivity groups. Identifiability can sometimes be improved by restricting the parameter space of $\alpha$ and $\beta$ according to the epidemiological context, but this often comes at the cost of flexibility. In Section~\ref{sec:age_mix_prob} we focus on the two component mixture for $T_i$ and extend it by allowing the mixing probabilities to vary with age and enable a more mechanistic representation of how seroreactivity evolves over the life course.

\section{Modelling age dependency}
\label{sec:age_dependency}

Age is a key determinant of antibody dynamics and must be incorporated explicitly into the latent distribution of $T_i$. We describe two complementary approaches. In the first, age is incorporated by allowing the shape parameters of the Beta distribution to vary with age. In the second, age is introduced through the mixing probabilities of a two component mixture model, where the functional form of the age effect can be guided by mechanistic models of antibody acquisition.

\subsection{Single Beta distribution with age-dependent parameters}
\label{sec:age_dependent_shapes}

In the simplest specification, the latent variable $T$ is modelled with a single Beta distribution whose shape parameters $(\alpha, \beta)$ vary with age. This formulation allows the density of $T$ to evolve smoothly across age, capturing the empirical observation that the distinction between individuals with and without an active immune response becomes progressively less pronounced as exposure accumulates in the population.

A convenient starting point is to use power functions of age, which we have found to work reasonably well in practice. Specifically, we write
\begin{align}
\label{eq:age_shapes_model}
    \alpha(a) &= \alpha_0 \, a^{\gamma}, \nonumber \\
    \beta(a)  &= \beta_0 \, a^{\delta},
\end{align}
where $\alpha_0, \beta_0 > 0$ and $\gamma, \delta$ are parameters to be estimated. This parameterization ensures that both the mean and the shape of the distribution adapt to different age groups. Empirically, for many infectious diseases, such as malaria, the proportion of individuals with detectable antibody levels tends to increase with cumulative exposure. However, antibody concentrations may also wane at older ages as immune responses decline or exposure becomes less frequent. Allowing the shape parameters of the Beta distribution to vary with age provides a flexible way to capture these gradual shifts in the antibody distribution, which may continue to evolve throughout life, albeit with diminishing intensity.

An alternative formulation to \eqref{eq:age_shapes_model} is to express the Beta distribution in terms of its \emph{mean--variance parameterization},  where
\begin{equation}
\label{eq:mean_var_model}
    \mathbb{E}(T) = \mu(a), \qquad 
\mathrm{Var}(T) = \frac{\mu(a)\,[1-\mu(a)]}{1+\phi}.
\end{equation}
Here, $\mu(a)\in(0,1)$ denotes the expected antibody level at age $a$, while the precision parameter $\phi>0$ determines how concentrated the distribution is around the mean. The corresponding shape parameters are, in this case, $\alpha(a) = \mu(a)\phi$ and $\beta(a) = [1-\mu(a)]\phi$. Hence, the standard GMM framework here is recovered when $\phi = 0$, as $T$ reduces to a Bernoulli variable such that $P[T = 1; a] = \mu(a)$. 

In this specification, we allow age to affect only the mean $\mu(a)$, while keeping the precision parameter $\phi$ constant across ages. This can result in more parsimonious models than~\eqref{eq:age_shapes_model}, and has the advantage that the effect of age can be interpreted directly on the mean of $T$. A convenient and interpretable choice is to model $\mu(a)$ through a logit-linear function of the log-transformed age:
\begin{equation}
\label{eq:logit-log-reg}
\log\left\{\frac{\xi(a)}{1-\xi(a)}\right\} = \eta_0 + \eta_1 \log(a).     
\end{equation}
We will show the application of this modelling approach in the case study of Section \ref{sec:applications}.

\subsection{Finite mixture models with age dependent mixing probabilities}
\label{sec:age_mix_prob}

An alternative way to introduce age dependence into the latent distribution of $T$ is through the mixing probabilities of a two-component Beta mixture model. In this formulation, the two Beta components represent subpopulations that differ in their propensity to mount an antibody response. The first component corresponds predominantly to individuals with low or absent immune activation, whereas the second represents those with a stronger or more sustained antibody response. This interpretation parallels that of classical finite mixture models directly fitted to the observed antibody measurements $Y$, but here it is formulated at the level of $T$ representing the latent seroreactivity level, hence the model interpretation differs.

To avoid label switching between the two Beta components, we impose the constraints $\alpha_1 < 1$ and $\beta_1 > 1$ for the first component, and $\alpha_2 > 1$ and $\beta_2 < 1$ for the second. This yields densities that are respectively concentrated near $T=0$ and $T=1$, providing a biologically interpretable distinction between the two subpopulations.

Let $\pi(a)$ denote the probability that an individual of age $a$ belongs to the second component related to higher levels of seroreactivity. The distribution of $T$ can then be expressed as
\begin{equation}
f(t ; a) = \bigl(1 - \pi(a)\bigr)\,\text{Beta}(t; \alpha_1,\beta_1)
+ \pi(a)\,\text{Beta}(t; \alpha_2,\beta_2).
\label{eq:age_mixture_lambda}
\end{equation}

Although our model operates on a continuous latent immune state $T$, it is still possible to parameterise the age dependent mixing weight $\pi(a)$ using mechanistic formulations inspired by the literature on analysis of serological data. These parameterisations allow biological processes such as exposure, boosting and waning to be incorporated in a parsimonious and interpretable way. In this latent formulation, $\pi(a)$ does not represent the probability of seroconversion in a binary sense. Instead, it governs the relative contribution of mixture components that place more or less mass on higher values of the continuous immune activation scale and therefore describes how measurable immune activity accumulates in the population with age.

A wide range of mechanistic models for the relation between age and antibody levels has been developed in the context of infectious disease research. These include catalytic models, models that allow repeated boosting through multiple exposures, and models with time varying hazards; a comprehensive overview of these approaches can be found in \citet{Hay2024Serodynamics}. Within our latent mixture framework, these formulations offer a convenient way to specify $\pi(a)$ while retaining epidemiological interpretability. For example, the catalytic framework assumes that the cumulative probability of activation follows
\begin{equation}
\label{eq:rev_cat_lambda}
\pi(a) = 1 - \exp\{-\lambda a\},    
\end{equation}
which arises from a constant hazard $\lambda$ for the onset of measurable immune activity. However, in our context, $\lambda$ should not be viewed as a seroconversion rate. Instead, it determines the speed at which higher levels of seroreactivity accumulate in the population as age increases. In other words, $\lambda$ controls how quickly the upper part of the latent immune state distribution gains weight with age, reflecting gradual changes in seroreactivity rather than transitions between discrete serostatus categories. We would expect the value of $\lambda$ to correlate with classical seroconversion rates, since both quantities reflect underlying exposure processes. Nevertheless, the values of $\lambda$ obtained from our latent mixture formulation are likely to be larger, because individuals can reach high levels of seroreactivity without necessarily becoming seropositive according to a predefined threshold in antibody levels.

\section{Inference}
\label{sec:inference}

Parameter estimation for the latent variable model requires integrating over the continuous distribution of the unobserved immune state $T$, which presents computational challenges for likelihood-based inference. We address this by considering and comparing two approaches: standard maximum likelihood estimation (MLE), and a computationally efficient histogram-based $L_2$ distance minimization method. This can be thought of a discretizations of the $L_2E$ method proposed by \cite{scott2001parametric}. Although our main motivation for using the method is computational gain, the limiting case i.e the $L_2E$ method is also used for it's robusntess \cite{scott2001parametric}. 

The observed data consist of antibody concentrations ${\bf Y}= [y_1,\ldots,y_n]$ measured across individuals. Under the conditional specification in \eqref{eq:lbg}, the marginal density of $Y$ is obtained by integrating over the latent variable distribution:
\begin{equation*}
f(y ; \vartheta)
\;=\;
\int_0^1 
\phi\!\left(y;\; (1-t)\mu_0+t\mu_1,\; (1-t)\sigma_0^2+t\sigma_1^2\right)
\, g_T(t;\psi)\, dt,
\end{equation*}
where $\phi(\cdot;\mu,\sigma^2)$ denotes the Gaussian density, $\vartheta =(\mu_0,\mu_1,\sigma_0,\sigma_1, \psi)$ encompasses all model parameters, and $\psi$ denotes the parameters governing the distribution of $T$ (e.g., Beta shape parameters $\alpha,\beta$). The exact log-likelihood is then
\begin{equation*}
\ell(\vartheta)
\;=\;
\sum_{i=1}^n \log f(y_i ; \vartheta).
\end{equation*}

Direct maximisation of $\ell$ requires repeated evaluation of the integral above, which is computationally expensive. To alleviate this burden, we adopt a histogram-based approximation. Let the observed ${\bf Y}$ values be summarised into a histogram with breakpoints $\{b_j\}_{j=0}^J$, bin widths $\Delta_j = b_j - b_{j-1}$, midpoints $m_j = (b_{j-1}+b_j)/2$, and bin counts $n_j$ with $\sum_{j=1}^J n_j = n$. The empirical density in bin $j$ is
$$
\widehat f_j \;=\; \frac{n_j}{n\,\Delta_j}, \qquad j=1,\ldots,J.
$$
On the model side, the bin probability is
$$
p_j(\vartheta)
\;=\;
\int_{b_{j-1}}^{b_j} f(y ; \vartheta)\,dy,
$$
which we approximate by the midpoint rule,
\begin{equation}
\label{eq:discr_dens}
p_j(\vartheta)
\;\approx\; f(m_j ; \vartheta)\,\Delta_j,
\end{equation}
with approximation error vanishing as $\Delta_j \to 0$.

Estimation is performed by comparing the empirical histogram $\widehat f_j$ with the model-implied densities $f(m_j ; \vartheta,\psi)$. Specifically, we minimise the $L_2$ criterion
\begin{equation}
\label{eq:$L_2$}
Q(\vartheta) 
\;=\; \sum_{j=1}^J  
\Bigl\{ \widehat f_j - f(m_j ; \vartheta) \Bigr\}^2,
\end{equation}
A logarithmic variant can also be defined based on $\widehat f_j \cdot log(f(m_j ; \vartheta,\psi))$, which corresponds to a discretization of the Kullback–Leibler divergence. An advantage of these approaches is that, unlike the exact likelihood, which requires evaluating the marginal distribution $f(y;\vartheta,\psi)$ at every observation, the histogram method requires evaluation only at the $J$ bin midpoints. Since typically $J \ll n$, this can lead to a substantial reduction in the number of integral computations. This is especially convenient when assessing the goodness of fit of several different models and also finding  starting values for MLE.

The theoretical justification for this approach rests on the fact that, as the bin widths $\Delta_j$ shrink, the histogram density $\widehat f_j$ converges to the empirical distribution of the data, while $f(m_j;\vartheta,\psi)$ converges to the model density at those points. In this limiting regime, minimisation of \eqref{eq:$L_2$} defines a minimum–distance estimator \citep{wolfowitz1957,beran1977} that minimises the $L_2$ distance between the empirical and model densities, i.e. the $L_2E$ method. 

We utilize the $L_2$ method for finding the initial estimates and input this as starting value for MLE in the case studies of Section \ref{sec:applications}. To quantify uncertainty in finite samples, we rely on a parametric bootstrap, where replicated datasets are generated under the fitted model and the full estimation procedure is repeated to approximate the sampling distribution of the estimators.

\subsection{Consistency of the histogram-based $L_2$ estimator}
\label{sec:$L_2$-consistency}

We establish the consistency of our estimator using the asymptotic framework of $M$-estimation \citep{vandervaartwellner1996}. The estimator is defined by minimizing the integrated squared error (ISE) between the histogram and the model, a method rooted in the $L_2E$ theory proposed by \cite{scott2001parametric}.

To align with the standard maximization notation of \cite{vandervaartwellner1996}, we define the estimator $\widehat{\vartheta}_n$ as the maximizer of the negative squared distance. Let $\mathbb{M}_n(\vartheta)$ denote the empirical criterion and $\mathbb{M}(\vartheta)$ the population criterion:
\begin{align*}
    \mathbb{M}_n(\vartheta) \;&:=\; - \int_I \Bigl\{\widehat{f}_n(y) - f_{n}(y; \vartheta)\Bigr\}^2 \, dy, \\
    \mathbb{M}(\vartheta) \;&:=\; - \int_I \Bigl\{f_0(y) - f(y; \vartheta)\Bigr\}^2 \, dy.
\end{align*}
Here, $\widehat{f}_n(y)$ is the histogram density estimator and $f_{n}(y; \vartheta)$ is the discretized (binned) model density in \eqref{eq:discr_dens}. The population criterion $\mathbb{M}(\vartheta)$ corresponds to the negative squared $L_2$ distance between the true density $f_0$ and the candidate model density $f(\cdot;\vartheta)$. As discussed by \cite{scott2001parametric}, maximizing $\mathbb{M}(\vartheta)$ is equivalent to minimizing the $L_2$ distance to the true data-generating distribution, which provides inherent robustness properties compared to maximum likelihood.

In our practical implementation, the integral is computed over a finite range covered by bins. For the theoretical proof, we treat the integral over the full support $I$ (which may be finite or infinite).

We first state the needed assumptions.
\begin{assumption}
\label{subsec:$L_2$-assumptions}

Let $\vartheta_0 \in \vartheta$ denote the true parameter such that $f_0(\cdot) = f(\cdot; \vartheta_0)$. We posit the following regularity conditions:

\begin{enumerate}[label=\textbf{A\arabic*.}, leftmargin=3em, itemsep=0.5ex]
    \item \label{asm:compact}
    \textbf{Compactness.} The parameter space $\vartheta$ is a compact subset of $\mathbb{R}^p$.
    
    \item \label{asm:ident}
    \textbf{Identification.} The true parameter $\vartheta_0$ is the unique maximizer of the population criterion $\mathbb{M}$ on $\vartheta$. Furthermore, for every open neighborhood $G$ of $\vartheta_0$, the maximum is well-separated:
    \begin{equation*}
        \mathbb{M}(\vartheta_0) \;>\; \sup_{\vartheta \in \vartheta \setminus G} \mathbb{M}(\vartheta).
    \end{equation*}
    
    \item \label{asm:regularity}
    \textbf{Regularity.} The model family satisfies $\sup_{\vartheta \in \vartheta} \|f(\cdot; \vartheta)\|_2 < \infty$. Additionally, the map $(y, \vartheta) \mapsto f(y; \vartheta)$ is measurable, and the true density $f_0 \in L_2$.
    
    \item \label{asm:histogram}
    \textbf{Histogram Regime.} The histogram partitions $\mathbb{R}$ into bins of equal width $h_n$. As $n \to \infty$, the bin width satisfies:
    \begin{equation*}
        h_n \to 0 \quad \text{and} \quad n h_n \to \infty.
    \end{equation*}
    
    \item \label{asm:smoothness}
    \textbf{Smoothness.} For every $\vartheta \in \vartheta$, the density $y \mapsto f(y;\vartheta)$ is absolutely continuous with first derivative $f'(y;\vartheta)$. We assume the derivative is square-integrable uniformly over the parameter space:
    \begin{equation*}
        \sup_{\vartheta \in \vartheta} \|f'(\cdot; \vartheta)\|_2 \;<\; \infty.
    \end{equation*}
\end{enumerate}
\end{assumption}
Under these assumptions we can establish establish the consistency of the proposed $L_2$ estimator.

\begin{theorem}
\label{lem:L_2-consistency}
Under Assumptions \ref{asm:compact}--\ref{asm:smoothness}, let $\widehat{\vartheta}_n$ be the maximizer of $\mathbb{M}_n$ then $\widehat{\vartheta}_n \xrightarrow{p} \vartheta_0$ as $n \to \infty$.
\end{theorem}

Proof of the above theorem is given in Section  \ref{sec:proof-theorem} of the Appendix. In Section \ref{subsec:regularity-verification} of the Appendix we prove that the set of Assumptions \ref{asm:compact}--\ref{asm:smoothness} hold for the proposed latent variable model in \eqref{eq:lbg} under mild conditions on the mixing distribution $g_T$.

\subsection{Model validation}
\label{sec:validation}

Model validation is based on graphical comparison of the empirical 
distribution of the data with replicate datasets simulated from the fitted model. In the simplest case, where age is not included as a covariate affecting either the Beta distribution of the latent variable $T_i$ or the mixing probabilities, this can be done by directly comparing the empirical histogram of the observed data with the histogram implied by the fitted model.  

When age is included, additional care is needed because the parameters of the latent distribution vary smoothly with age. In this situation, a single global comparison between the empirical histogram and the fitted marginal distribution of $Y$ is not informative, since any apparent lack of fit may arise from discrepancies at specific ages rather than from a systematic problem across all ages. To address this, we validate the model by constructing replicate datasets that preserve the observed age structure of the sample. Conceptually, this means simulating outcomes conditional on the observed ages and then assessing model adequacy by comparing histograms stratified into age groups. More specifically, the procedure can be summarised as follows.

\begin{enumerate}
    \item For each observed individual (or within each narrow age band), compute the corresponding age-specific parameters of the latent distribution of $T$ using the fitted model.  

    \item Draw a value of $T$ from the corresponding age-specific distribution.  

    \item Conditional on $T$, generate an antibody concentration $Y$ from the Gaussian distribution in \eqref{eq:lbg}.  

    \item Repeat steps 1–3 for all individuals to obtain a single replicate dataset under the fitted model.  

    \item Repeat the entire process many times to generate a collection of replicate datasets.  

    \item For each replicate, compute the empirical histogram of the simulated antibody concentrations, either overall or within age groups.  
\end{enumerate}

Comparison is then made between the empirical histogram of the observed data and the collection of histograms obtained from the simulated datasets. An envelope can be formed from the simulations (for example, by plotting the pointwise range across replicates), providing a visual check of whether the observed histogram is consistent with the fitted model. If the observed histogram lies within the simulated envelope, this indicates that the model 
provides an adequate representation of the observed data.

\section{Simulation studies}
\label{sec:simulation}

\subsection{Assessment of the $L_2$ estimator performance}
\label{subsec:sim_study_L2}

To evaluate the performance of the $L_2$ histogram-based estimator relative to MLE, we conducted a comprehensive simulation study under different scenarios that mimic specific epidemiological settings. The simulation study aims to compare the two methods in terms of both statistical efficiency (bias and root mean squared error) and computational efficiency (computation time).

For each simulation replicate, we generated data from the latent Beta  model specified in \eqref{eq:lbg}. Specifically, for each individual $i = 1, \ldots, n$, we sampled the latent immune state from a Beta distribution, $T_i \sim \text{Beta}(\alpha, \beta)$, and then, conditional on $T_i = t_i$, sampled the observed antibody concentration from
the Gaussian distribution of $Y_i \mid T_i = t_i$.

We considered four scenarios representing distinct epidemiological settings, chosen to reflect the range of serological patterns commonly observed in practice. Each scenario is defined by a specific combination of parameters $(\mu_0, \mu_1, \sigma_0, \sigma_1, \alpha, \beta)$, with the Beta shape parameters $(\alpha, \beta)$ determining the distribution of latent immune states in the population. 

\begin{itemize}
    \item The first scenario, \emph{bimodal} (BM) ($\alpha = 0.5$, $\beta = 0.5$), uses a strongly U-shaped Beta distribution with equal probability mass near both extremes, corresponding to highly polarized populations with clear separation between seronegative and seropositive groups, with $\mathbb{E}[T] = 0.5$.
    
    \item The second scenario, \emph{high transmission} (HT) ($\alpha = 3.0$, $\beta = 0.5$), uses a U-shaped Beta distribution concentrated near $T = 1$, characteristic of endemic transmission where repeated exposure is common; most individuals show high seroreactivity, with $\mathbb{E}[T] = 0.86$. 
    
    \item The third scenario, \emph{intermediate transmission} (IT) ($\alpha = 2.0$, $\beta = 2.0$), employs a symmetric unimodal Beta distribution with $\mathbb{E}[T] = 0.5$, representing populations with balanced exposure patterns and capturing settings with moderate, stable transmission or heterogeneous exposure across subgroups. 
    
    \item The fourth scenario, \emph{low transmission} (LT) ($\alpha = 0.5$, $\beta = 3.0$), uses a U-shaped Beta distribution concentrated near $T = 0$, representing populations with minimal ongoing pathogen exposure, such as post-elimination settings or areas with effective vector control. Under this configuration, the majority of individuals exhibit low seroreactivity, with $\mathbb{E}[T] = 0.14$. 
\end{itemize}

For all scenarios, we fixed the conditional distribution parameters at $\mu_0 = -3.0$, $\mu_1 = 1.0$, $\sigma_0 = 0.8$, and $\sigma_1 = 0.3$, except for the bimodal scenario where we used $\mu_0 = -3.5$, $\mu_1 = 1.5$, $\sigma_0 = 0.7$, and $\sigma_1 = 0.2$ to increase the separation between low and high seroreactivity states. These parameter values were chosen to reflect realistic antibody concentration ranges on the log scale, with $\mu_1 - \mu_0 = 4$ (or 4.5 for the bimodal scenario) providing sufficient separation to distinguish the two extremes while $\sigma_1 < \sigma_0$ captures the reduced variability typically observed near assay saturation.

We evaluated performance across four sample sizes: $n \in \{100, 500, 1000, 5000\}$. This range enables us to assess how the relative performance of $L_2$ and MLE changes with increasing information. For each combination of scenario and sample size, we generated 1,000 independent replicates. For each simulated dataset, we applied both the $L_2$ histogram-based estimator and MLE, recording parameter estimates and computation times for both methods. We construct histograms using Sturges' rule to determine the number of bins: $k = \lceil \log_2(n) + 1 \rceil$. This binning rule yields between 8 bins ($n = 100$) and 14 bins ($n = 5000$).

For the estimates of each of the model parameters performance was evaluated using the median bias and the median absolute error. To assess the computational efficiency, we recorded the elapsed wall-clock time
(in seconds) required to obtain parameter estimates for each method on every simulation replicate.

\subsubsection{Results}
\label{sec:sim_results}

\begin{table}[ht!]
\centering
\footnotesize
\caption{Median bias of parameter estimates across scenarios and sample sizes.}
\label{tab:bias_all}
\begin{tabular}{llrrrrrrrr}
\toprule
Scenario & Parameter & \multicolumn{2}{c}{$n=100$} & \multicolumn{2}{c}{$n=500$} & \multicolumn{2}{c}{$n=1000$} & \multicolumn{2}{c}{$n=5000$} \\
& & $L_2$ & MLE & $L_2$ & MLE & $L_2$ & MLE & $L_2$ & MLE \\
\midrule
BM & $\alpha$ & $-0.105$ & $-0.054$ & $-0.026$ & $-0.008$ & $-0.023$ & $-0.001$ & $-0.005$ & $-0.000$ \\
 & $\beta$ & $-0.073$ & $-0.038$ & $-0.028$ & $-0.006$ & $-0.025$ & $-0.001$ & $-0.007$ & $-0.002$ \\
 & $\mu_0$ & $0.112$ & $0.064$ & $0.055$ & $0.014$ & $0.038$ & $-0.002$ & $0.006$ & $-0.001$ \\
 & $\mu_1$ & $-0.118$ & $-0.022$ & $-0.028$ & $-0.002$ & $-0.017$ & $-0.001$ & $-0.001$ & $-0.000$ \\
 & $\sigma_0$ & $0.032$ & $-0.044$ & $0.036$ & $-0.004$ & $0.031$ & $-0.009$ & $0.015$ & $-0.001$ \\
 & $\sigma_1$ & $0.060$ & $-0.023$ & $0.084$ & $-0.004$ & $0.076$ & $-0.003$ & $0.061$ & $-0.000$ \\
\midrule
HT & $\alpha$ & $-2.105$ & $-2.026$ & $-1.320$ & $-0.494$ & $-0.742$ & $-0.113$ & $-0.027$ & $0.269$ \\
 & $\beta$ & $-0.286$ & $-0.214$ & $-0.109$ & $-0.071$ & $-0.020$ & $-0.036$ & $0.004$ & $0.012$ \\
 & $\mu_0$ & $1.539$ & $0.995$ & $1.062$ & $0.109$ & $0.806$ & $-0.198$ & $0.120$ & $-0.295$ \\
 & $\mu_1$ & $-0.052$ & $-0.055$ & $-0.053$ & $-0.019$ & $-0.036$ & $-0.005$ & $-0.009$ & $0.011$ \\
 & $\sigma_0$ & $-0.360$ & $-0.579$ & $0.090$ & $-0.115$ & $0.142$ & $-0.164$ & $0.086$ & $-0.061$ \\
 & $\sigma_1$ & $0.013$ & $-0.015$ & $0.022$ & $-0.005$ & $0.016$ & $-0.002$ & $0.019$ & $-0.003$ \\
\midrule
IT & $\alpha$ & $-1.826$ & $-1.751$ & $-1.766$ & $-1.534$ & $-1.660$ & $-1.174$ & $-0.850$ & $-0.239$ \\
 & $\beta$ & $-1.705$ & $-1.523$ & $-1.612$ & $-1.124$ & $-1.322$ & $-0.696$ & $-0.469$ & $-0.071$ \\
 & $\mu_0$ & $1.335$ & $1.325$ & $1.174$ & $1.043$ & $1.101$ & $0.873$ & $0.685$ & $0.167$ \\
 & $\mu_1$ & $-0.747$ & $-0.651$ & $-0.698$ & $-0.396$ & $-0.581$ & $-0.283$ & $-0.234$ & $-0.029$ \\
 & $\sigma_0$ & $0.037$ & $0.028$ & $0.132$ & $0.081$ & $0.124$ & $0.072$ & $0.094$ & $0.026$ \\
 & $\sigma_1$ & $-0.036$ & $-0.108$ & $-0.018$ & $-0.131$ & $0.062$ & $-0.059$ & $0.034$ & $-0.022$ \\
\midrule
LT & $\alpha$ & $-0.444$ & $-0.443$ & $-0.442$ & $-0.411$ & $-0.422$ & $-0.328$ & $-0.306$ & $-0.224$ \\
 & $\beta$ & $-2.797$ & $-2.605$ & $-2.828$ & $-2.531$ & $-2.730$ & $-2.347$ & $-2.274$ & $-1.794$ \\
 & $\mu_0$ & $0.232$ & $0.315$ & $0.283$ & $0.291$ & $0.269$ & $0.218$ & $0.194$ & $0.143$ \\
 & $\mu_1$ & $-1.777$ & $-1.569$ & $-1.681$ & $-1.426$ & $-1.481$ & $-1.304$ & $-1.403$ & $-1.218$ \\
 & $\sigma_0$ & $-0.026$ & $0.008$ & $0.025$ & $0.027$ & $0.030$ & $0.022$ & $0.030$ & $0.015$ \\
 & $\sigma_1$ & $-0.030$ & $-0.110$ & $0.204$ & $-0.009$ & $0.259$ & $0.137$ & $0.364$ & $0.259$ \\
\bottomrule
\end{tabular}
\end{table}

\begin{table}[ht!]
\centering
\footnotesize
\caption{Median Absolute Error (MAE) of parameter estimates across scenarios and sample sizes.}
\label{tab:mae_all}
\begin{tabular}{llrrrrrrrr}
\toprule
Scenario & Parameter & \multicolumn{2}{c}{$n=100$} & \multicolumn{2}{c}{$n=500$} & \multicolumn{2}{c}{$n=1000$} & \multicolumn{2}{c}{$n=5000$} \\
& & $L_2$ & MLE & $L_2$ & MLE & $L_2$ & MLE & $L_2$ & MLE \\
\midrule
BM & $\alpha$ & $0.105$ & $0.054$ & $0.026$ & $0.008$ & $0.023$ & $0.001$ & $0.005$ & $0.000$ \\
 & $\beta$ & $0.073$ & $0.038$ & $0.028$ & $0.006$ & $0.025$ & $0.001$ & $0.007$ & $0.002$ \\
 & $\mu_0$ & $0.112$ & $0.064$ & $0.055$ & $0.014$ & $0.038$ & $0.002$ & $0.006$ & $0.001$ \\
 & $\mu_1$ & $0.118$ & $0.022$ & $0.028$ & $0.002$ & $0.017$ & $0.001$ & $0.001$ & $0.000$ \\
 & $\sigma_0$ & $0.032$ & $0.044$ & $0.036$ & $0.004$ & $0.031$ & $0.009$ & $0.015$ & $0.001$ \\
 & $\sigma_1$ & $0.060$ & $0.023$ & $0.084$ & $0.004$ & $0.076$ & $0.003$ & $0.061$ & $0.000$ \\
\midrule
HT & $\alpha$ & $2.105$ & $2.026$ & $1.320$ & $0.494$ & $0.742$ & $0.113$ & $0.027$ & $0.269$ \\
 & $\beta$ & $0.286$ & $0.214$ & $0.109$ & $0.071$ & $0.020$ & $0.036$ & $0.004$ & $0.012$ \\
 & $\mu_0$ & $1.539$ & $0.995$ & $1.062$ & $0.109$ & $0.806$ & $0.198$ & $0.120$ & $0.295$ \\
 & $\mu_1$ & $0.052$ & $0.055$ & $0.053$ & $0.019$ & $0.036$ & $0.005$ & $0.009$ & $0.011$ \\
 & $\sigma_0$ & $0.360$ & $0.579$ & $0.090$ & $0.115$ & $0.142$ & $0.164$ & $0.086$ & $0.061$ \\
 & $\sigma_1$ & $0.013$ & $0.015$ & $0.022$ & $0.005$ & $0.016$ & $0.002$ & $0.019$ & $0.003$ \\
\midrule
IT & $\alpha$ & $1.826$ & $1.751$ & $1.766$ & $1.534$ & $1.660$ & $1.174$ & $0.850$ & $0.239$ \\
 & $\beta$ & $1.705$ & $1.523$ & $1.612$ & $1.124$ & $1.322$ & $0.696$ & $0.469$ & $0.071$ \\
 & $\mu_0$ & $1.335$ & $1.325$ & $1.174$ & $1.043$ & $1.101$ & $0.873$ & $0.685$ & $0.167$ \\
 & $\mu_1$ & $0.747$ & $0.651$ & $0.698$ & $0.396$ & $0.581$ & $0.283$ & $0.234$ & $0.029$ \\
 & $\sigma_0$ & $0.037$ & $0.028$ & $0.132$ & $0.081$ & $0.124$ & $0.072$ & $0.094$ & $0.026$ \\
 & $\sigma_1$ & $0.036$ & $0.108$ & $0.018$ & $0.131$ & $0.062$ & $0.059$ & $0.034$ & $0.022$ \\
\midrule
LT & $\alpha$ & $0.444$ & $0.443$ & $0.442$ & $0.411$ & $0.422$ & $0.328$ & $0.306$ & $0.224$ \\
 & $\beta$ & $2.797$ & $2.605$ & $2.828$ & $2.531$ & $2.730$ & $2.347$ & $2.274$ & $1.794$ \\
 & $\mu_0$ & $0.232$ & $0.315$ & $0.283$ & $0.291$ & $0.269$ & $0.218$ & $0.194$ & $0.143$ \\
 & $\mu_1$ & $1.777$ & $1.569$ & $1.681$ & $1.426$ & $1.481$ & $1.304$ & $1.403$ & $1.218$ \\
 & $\sigma_0$ & $0.026$ & $0.008$ & $0.025$ & $0.027$ & $0.030$ & $0.022$ & $0.030$ & $0.015$ \\
 & $\sigma_1$ & $0.030$ & $0.110$ & $0.204$ & $0.009$ & $0.259$ & $0.137$ & $0.364$ & $0.259$ \\
\bottomrule
\end{tabular}
\end{table}

\begin{table}[ht!]
\centering
\small
\caption{Mean computation time (seconds) for $L_2$ and MLE estimation across scenarios and sample sizes.}
\label{tab:time_all}
\begin{tabular}{lrrrrrrrr}
\toprule
Scenario & \multicolumn{2}{c}{$n=100$} & \multicolumn{2}{c}{$n=500$} & \multicolumn{2}{c}{$n=1000$} & \multicolumn{2}{c}{$n=5000$} \\
 & $L_2$ & MLE & $L_2$ & MLE & $L_2$ & MLE & $L_2$ & MLE \\
\midrule
BM & 0.780 & 6.690 & 0.640 & 21.590 & 0.660 & 38.150 & 0.790 & 177.270 \\
HT & 1.510 & 14.320 & 1.690 & 68.750 & 1.630 & 117.780 & 1.540 & 396.950 \\
IT & 1.500 & 11.800 & 2.050 & 63.610 & 2.020 & 120.610 & 1.770 & 451.080 \\
LT & 1.780 & 12.130 & 2.540 & 66.130 & 2.810 & 133.630 & 3.590 & 687.420 \\
\bottomrule
\end{tabular}
\end{table}

Tables~\ref{tab:bias_all}, \ref{tab:mae_all}, and \ref{tab:time_all} present the simulation results for median bias, MAE, and computation time, respectively. 

The BM scenario exhibits the best overall performance for both estimation methods, with median bias and MAE approaching zero even at moderate sample sizes. At $n=1000$, both $L_2$ and MLE achieve negligible bias and MAE (all below 0.03) for all parameters. The HT scenario shows large biases at small sample sizes but rapid improvement with increasing $n$, with MLE reaching near-zero bias by $n=5000$ while $L_2$ converges more slowly. In contrast, the IT and LT scenarios prove substantially more challenging. For IT, both methods exhibit persistent negative bias in $\alpha$ and $\beta$ even at $n=5000$ (e.g., $L_2$: $-0.850$ and $-0.469$; MLE: $-0.239$ and $-0.071$). The LT scenario emerges as the most challenging, with median bias in $\beta$ exceeding 1.7 in absolute value for both methods even at the largest sample size. These difficulties arise because symmetric (IT) and highly skewed (LT) Beta distributions are harder to distinguish from the marginal distribution of $Y$, leading to weaker empirical identifiability.

Examining performance across parameters reveals distinct patterns in the relative performance of $L_2$ versus MLE. For the Beta shape parameters $\alpha$ and $\beta$, MLE consistently outperforms $L_2$ in terms of MAE, particularly in the HT and IT scenarios where MLE achieves 2 to 7 times lower MAE at $n=5000$. For the conditional mean parameters $\mu_0$ and $\mu_1$, the two methods perform more comparably in the BM scenario, while MLE shows clear advantages in the more challenging scenarios, especially for $\mu_1$ where the relative MAE difference can exceed an order of magnitude at large sample sizes (e.g., IT scenario at $n=5000$: 0.234 vs 0.029). The variance parameters $\sigma_0$ and $\sigma_1$  exhibit mixed results, with $L_2$ occasionally showing smaller bias than MLE in the BM scenario, but MLE generally achieving lower MAE in more challenging scenarios. Overall, MLE demonstrates asymptotic superiority across all parameters, though the practical advantage varies considerably by scenario and sample size, with the gap being smallest for BM and largest for IT and LT.

The computational efficiency advantage of the $L_2$ method is substantial and increases dramatically with sample size (Table~\ref{tab:time_all}). At $n=100$, $L_2$ is approximately 7--10 times faster than MLE across scenarios. This speedup factor increases markedly with sample size, reaching 26--41-fold at $n=500$, 48--72-fold at $n=1000$, and 192--258-fold at $n=5000$. Even in the most challenging LT scenario, $L_2$ estimation at $n=5000$ requires only about 3.6 seconds, on average, compared to 687 seconds for MLE—a 191-fold speedup that becomes critical when fitting models repeatedly during exploratory analysis and bootstrap inference.

\subsection{Assessment of the impact of model misspesification}
\label{subsec:sim_study_miss}

In this second simulation study, we assess the robustness of maximum likelihood estimation when the data-generating mechanism departs from the fitted model specification. Specifically, we consider three scenarios in which subpopulations experience different patterns of pathogen exposure that are not accounted for in the standard age-dependent structure.

The fitted model in all scenarios is the latent variable model with mean-variance parameterization of the latent Beta distribution as given in equation~\eqref{eq:mean_var_model}, where the mean $\mu(a)$ depends on age through the logit-linear specification in~\eqref{eq:logit-log-reg}. The true parameters governing the baseline age-dependent distribution are $\eta_0 = -2.0$, $\eta_1 = 0.4$, and $\phi = 3.0$ for the distribution of $T$, and $\mu_0 = -3.5$, $\mu_1 = 1.0$, $\sigma_0 = 0.8$, and $\sigma_1 = 0.3$ for the conditional distribution of $Y$ given $T$. Rather than sampling directly from this model, we introduce structured heterogeneity representing plausible departures from the assumed age-dependent dynamics. For each scenario, we consider sample sizes $n \in \{1000, 2000, 5000\}$ and generate 2000 simulated data-sets.

The first scenario represents populations in which a proportion of individuals have received interventions that directly elevate antibody levels through mechanisms distinct from natural infection. Examples include prophylactic antibody therapy or vaccines that produce measurable serological responses that would be reflected in elevated values of the latent ser-reactivity state. Hence, we simulate under this scenario by first sampling each individual's baseline latent state $T_i$ from the age-dependent Beta distribution, then adding a stochastic boost $\Delta_i \sim \text{Beta}(3,3)$ to a randomly selected subset representing vaccinated individuals. The boosted latent state is given by $T_i^{*} = \min(1, T_i + \Delta_i)$, where the truncation at unity reflects the upper bound of the latent scale. We shall use the parameter $\rho \in \{0.2, 0.5, 0.8\}$ to denote the coverage of the intervention, assumed to be uniform across all age groups.

The second scenario represents age-targeted public health interventions that indirectly reduce pathogen exposure for specific age groups. Examples include school-based health programs, nutritional supplementation initiatives, or targeted vector control efforts that primarily benefit children of school age. Rather than directly boosting antibody levels, these interventions reduce the force of infection, thereby dampening the age-dependent accumulation of seroreactivity. We implement this by modifying the expected latent seroreactivity for individuals aged between 5 and 15 years according to 
\begin{equation*}
\mu^*(a) = \frac{\mu(a) \exp(-\delta)}{\mu(a) \exp(-\delta) + 1 - \mu(a)},
\end{equation*}
where $\mu(a)$ denotes the age-specific mean under the standard model and $\delta \in \{0.4, 0.7, 1.2\}$ quantifies the strength of exposure reduction. Larger values of $\delta$ correspond to greater reductions in the expected seroreactivity for the targeted age range. For ages outside the interval $[5,15]$, we set $\mu^*(a) = \mu(a)$. The latent state for each individual is then sampled from $T_i \sim \text{Beta}(\mu^*(a_i)\phi, [1-\mu^*(a_i)]\phi)$.

The third scenario represents immigration or population movement that introduces individuals with no prior pathogen exposure into endemic areas. In cross-sectional serological surveys, such individuals may be sampled before they have had sufficient contact with the pathogen to mount a detectable immune response, resulting in an absence of measurable seroreactivity at the time of sampling. We model this by designating a proportion $r \in \{0.10, 0.25, 0.40\}$ of individuals as completely unexposed, assigning them $T_i = 0$ exactly, while the remaining $1 - r$ proportion are sampled from the age-dependent Beta distribution. This creates a discrete point mass at zero seroreactivity superimposed on the continuous age-dependent distribution, representing a fundamentally different form of population heterogeneity than the model assumes. 

After fitting the mean-variance Beta model to each of the 2,000 simulated datasets, we compute the median bias and the median absolute deviation  for parameters of the conditional distribution of $Y_i$, namely $\mu_0$, $\mu_1$, $\sigma^2_0$ and $\sigma^2_1$, and the parameters of the latent variable $T_i$, consisting of $\eta_0$, $\eta_1$ and the variance parameter $\phi$.

Although we monitor estimation performance across all seven model parameters, our primary interest lies in the recovery of $\eta_0$ and $\eta_1$, which govern the age-dependent mean of the latent seroreactivity distribution and therefore determine the epidemiological inferences drawn from the fitted model. 

\subsubsection{Results}
\begin{table}[ht!]
\centering
\caption{Scenario 1: Intervention-induced antibody elevation. Parameter $\rho$ represents intervention coverage. Median bias and median absolute error (MAE) are reported for all seven model parameters across sample sizes and coverage levels. Results based on 2000 simulation replicates per configuration.}
\label{tab:scenario1_results}
\begin{tabular}{llrrrrrrr}
\toprule
Parameter & $n$ & \multicolumn{2}{c}{$\rho=0.2$} & \multicolumn{2}{c}{$\rho=0.5$} & \multicolumn{2}{c}{$\rho=0.8$} \\
 & & Bias & MAE & Bias & MAE & Bias & MAE \\
\midrule
$\mu_0$ & 1000 & 0.416 & 0.542 & 0.510 & 0.589 & 0.618 & 0.670 \\
 & 2000 & 0.276 & 0.367 & 0.383 & 0.426 & 0.545 & 0.558 \\
 & 5000 & 0.227 & 0.271 & 0.352 & 0.361 & 0.522 & 0.523 \\
\midrule
$\mu_1$ & 1000 & -0.105 & 0.569 & 0.136 & 0.467 & 0.220 & 0.387 \\
 & 2000 & 0.129 & 0.386 & 0.289 & 0.386 & 0.299 & 0.349 \\
 & 5000 & 0.175 & 0.260 & 0.317 & 0.334 & 0.327 & 0.334 \\
\midrule
$\sigma_0$ & 1000 & 0.065 & 0.113 & 0.079 & 0.119 & 0.084 & 0.119 \\
 & 2000 & 0.045 & 0.079 & 0.063 & 0.084 & 0.076 & 0.091 \\
 & 5000 & 0.039 & 0.053 & 0.059 & 0.065 & 0.075 & 0.076 \\
\midrule
$\sigma_1$ & 1000 & -0.015 & 0.300 & -0.133 & 0.300 & -0.192 & 0.300 \\
 & 2000 & -0.078 & 0.300 & -0.176 & 0.300 & -0.184 & 0.255 \\
 & 5000 & -0.068 & 0.133 & -0.145 & 0.161 & -0.164 & 0.170 \\
\midrule
$\eta_0$ & 1000 & -0.041 & 0.594 & 0.021 & 0.489 & 0.079 & 0.417 \\
 & 2000 & 0.099 & 0.326 & 0.138 & 0.276 & 0.143 & 0.248 \\
 & 5000 & 0.152 & 0.194 & 0.162 & 0.182 & 0.166 & 0.183 \\
\midrule
$\eta_1$ & 1000 & -0.104 & 0.135 & -0.096 & 0.114 & -0.077 & 0.097 \\
 & 2000 & -0.084 & 0.089 & -0.078 & 0.082 & -0.067 & 0.071 \\
 & 5000 & -0.070 & 0.070 & -0.069 & 0.069 & -0.064 & 0.064 \\
\midrule
$\phi$ & 1000 & -0.683 & 1.352 & -0.595 & 0.994 & -0.607 & 0.839 \\
 & 2000 & -0.227 & 0.693 & -0.234 & 0.531 & -0.429 & 0.547 \\
 & 5000 & -0.083 & 0.386 & -0.137 & 0.300 & -0.335 & 0.376 \\
\bottomrule
\end{tabular}
\end{table}

Under Scenario 1 (Table \ref{tab:scenario1_results}), the most notable consequence of the intervention-induced boost is a systematic upward displacement of $\mu_0$, the lower boundary parameter of the conditional distribution. The median bias in $\mu_0$ increases monotonically with coverage $\rho$, ranging from $0.42$ at $\rho = 0.2$ to $0.62$ at $\rho = 0.8$ for $n = 1000$, and remaining non-negligible even at $n = 5000$.  Importantly, however, the age-trend parameters $\eta_0$ and $\eta_1$ show small median bias across all sample sizes and $\rho$ values, with median biases not exceeding $0.17$ and $0.10$ in absolute value respectively at $n = 5000$. The precision parameter $\phi$ shows moderate negative bias at small sample sizes that largely decreases by $n = 5000$. The parameter $\sigma_1$ exhibits MAE of around $0.3$ at $n = 1000$ and $n = 2000$, indicating that the high-seroreactivity variance becomes poorly identified when the upper tail of the antibody distribution is distorted by the boost, though this also decreases at $n = 5000$. Overall, the observed patterns can be explained by the fact that the random boosting process shifts the observed antibody distribution upward, and the model largely absorbs part of this shift into the lower boundary estimate of the conditional model of $Y_i$ rather than into the age-dependent latent structure of $T_i$. 

\begin{table}[ht!]
\centering
\caption{Scenario 2: Age-targeted exposure reduction for individuals aged 5--15 years. Parameter $\delta$ represents the strength of exposure dampening on the logit scale for the age group 5--15 years. Median bias and MAE reported for all seven model parameters. Results based on 2000 simulation replicates per configuration.}
\label{tab:scenario2_results}
\begin{tabular}{llrrrrrrr}
\toprule
Parameter & $n$ & \multicolumn{2}{c}{$\delta=0.4$} & \multicolumn{2}{c}{$\delta=0.7$} & \multicolumn{2}{c}{$\delta=1.2$} \\
 & & Bias & MAE & Bias & MAE & Bias & MAE \\
\midrule
$\mu_0$ & 1000 & 0.410 & 0.456 & 0.371 & 0.408 & 0.299 & 0.323 \\
 & 2000 & 0.356 & 0.375 & 0.334 & 0.347 & 0.268 & 0.276 \\
 & 5000 & 0.312 & 0.318 & 0.307 & 0.312 & 0.255 & 0.258 \\
\midrule
$\mu_1$ & 1000 & -0.688 & 0.890 & -0.755 & 0.968 & -0.626 & 0.902 \\
 & 2000 & -0.293 & 0.375 & -0.283 & 0.354 & -0.282 & 0.348 \\
 & 5000 & -0.205 & 0.242 & -0.218 & 0.244 & -0.193 & 0.230 \\
\midrule
$\sigma_0$ & 1000 & 0.075 & 0.095 & 0.074 & 0.090 & 0.072 & 0.085 \\
 & 2000 & 0.071 & 0.081 & 0.077 & 0.083 & 0.074 & 0.079 \\
 & 5000 & 0.066 & 0.068 & 0.075 & 0.077 & 0.073 & 0.075 \\
\midrule
$\sigma_1$ & 1000 & 0.181 & 0.300 & 0.214 & 0.300 & 0.179 & 0.300 \\
 & 2000 & 0.061 & 0.300 & 0.062 & 0.300 & 0.064 & 0.300 \\
 & 5000 & 0.042 & 0.114 & 0.047 & 0.123 & 0.043 & 0.116 \\
\midrule
$\eta_0$ & 1000 & -0.972 & 1.116 & -1.463 & 1.543 & -1.804 & 1.827 \\
 & 2000 & -0.612 & 0.643 & -1.003 & 1.017 & -1.451 & 1.451 \\
 & 5000 & -0.469 & 0.474 & -0.866 & 0.871 & -1.324 & 1.324 \\
\midrule
$\eta_1$ & 1000 & 0.023 & 0.292 & 0.134 & 0.442 & 0.270 & 0.517 \\
 & 2000 & 0.070 & 0.154 & 0.177 & 0.256 & 0.311 & 0.391 \\
 & 5000 & 0.092 & 0.110 & 0.201 & 0.219 & 0.337 & 0.355 \\
\midrule
$\phi$ & 1000 & -1.316 & 2.311 & -1.380 & 2.566 & -1.151 & 2.540 \\
 & 2000 & -0.749 & 0.842 & -0.733 & 0.806 & -0.671 & 0.734 \\
 & 5000 & -0.588 & 0.605 & -0.619 & 0.640 & -0.548 & 0.570 \\
\bottomrule
\end{tabular}
\end{table}

Scenario 2 produces the most severe and structurally persistent misspecification of the three scenarios. The suppression of seroreactivity among individuals aged 5 to 15 years creates a local dip in the age-antibody profile that the fitted logit-linear model in age cannot represent, resulting in a systematic distortion of the age-trend parameters. Specifically, $\eta_0$ shows large negative median bias that worsens monotonically with $\delta$ and fails to diminish meaningfully with increasing sample size: at $n = 5000$, the median bias reaches $-0.47$, $-0.87$, and $-1.32$ for $\delta = 0.4$, $0.7$, and $1.2$ respectively, corresponding to relative errors of up to $66\%$ on a parameter whose true value is $-2.0$. The parameter $\eta_1$ 
acquires positive median bias that also increases with $\delta$ and persists at $n = 5000$ ($0.09$, $0.20$, $0.34$), indicating that the model compensates for the artificially low intercept by steepening the estimated age gradient. This compensation mechanism is internally coherent but epidemiologically misleading: the fitted model would suggest a 
steeper accumulation of seroreactivity with age than is actually present in the unaffected age groups. The lower boundary $\mu_0$ also acquires persistent positive bias of approximately $0.3$ that does not diminish with $n$, while $\phi$ shows sustained negative bias reflecting reduced apparent precision when the latent age structure is misspecified. The non-vanishing nature of the biases in $\eta_0$ and $\eta_1$ as $n$ increases confirms that this scenario induces genuine structural misspecification that cannot be resolved through larger sample sizes alone.

\begin{table}[ht!]
\centering
\caption{Scenario 3: Introduction of immunologically naive subpopulation. Parameter $r$ represents the proportion of unexposed individuals. Median bias and MAE reported for all seven model parameters.  Results based on 2000 simulation replicates per configuration.}
\label{tab:scenario3_results}
\begin{tabular}{llrrrrrrr}
\toprule
Parameter & $n$ & \multicolumn{2}{c}{$r=0.10$} & \multicolumn{2}{c}{$r=0.25$} & \multicolumn{2}{c}{$r=0.40$} \\
 & & Bias & MAE & Bias & MAE & Bias & MAE \\
\midrule
$\mu_0$ & 1000 & 0.309 & 0.392 & 0.183 & 0.246 & 0.091 & 0.156 \\
 & 2000 & 0.232 & 0.293 & 0.154 & 0.190 & 0.066 & 0.120 \\
 & 5000 & 0.188 & 0.202 & 0.125 & 0.135 & 0.056 & 0.083 \\
\midrule
$\mu_1$ & 1000 & -0.635 & 0.752 & -1.013 & 1.068 & -1.229 & 1.270 \\
 & 2000 & -0.331 & 0.415 & -0.653 & 0.675 & -0.800 & 0.802 \\
 & 5000 & -0.282 & 0.314 & -0.587 & 0.589 & -0.732 & 0.732 \\
\midrule
$\sigma_0$ & 1000 & 0.075 & 0.101 & 0.057 & 0.079 & 0.036 & 0.059 \\
 & 2000 & 0.067 & 0.081 & 0.061 & 0.068 & 0.036 & 0.045 \\
 & 5000 & 0.062 & 0.064 & 0.056 & 0.057 & 0.036 & 0.038 \\
\midrule
$\sigma_1$ & 1000 & 0.177 & 0.300 & 0.296 & 0.300 & 0.369 & 0.369 \\
 & 2000 & 0.081 & 0.300 & 0.203 & 0.300 & 0.246 & 0.300 \\
 & 5000 & 0.076 & 0.137 & 0.189 & 0.237 & 0.243 & 0.294 \\
\midrule
$\eta_0$ & 1000 & -0.118 & 0.583 & -0.086 & 0.555 & -0.190 & 0.513 \\
 & 2000 & 0.072 & 0.308 & 0.082 & 0.328 & 0.009 & 0.296 \\
 & 5000 & 0.133 & 0.190 & 0.185 & 0.226 & 0.115 & 0.192 \\
\midrule
$\eta_1$ & 1000 & -0.115 & 0.156 & -0.157 & 0.190 & -0.192 & 0.225 \\
 & 2000 & -0.097 & 0.104 & -0.129 & 0.137 & -0.161 & 0.167 \\
 & 5000 & -0.077 & 0.077 & -0.099 & 0.099 & -0.124 & 0.126 \\
\midrule
$\phi$ & 1000 & -1.270 & 1.502 & -1.741 & 1.866 & -1.994 & 2.053 \\
 & 2000 & -0.885 & 1.006 & -1.415 & 1.435 & -1.636 & 1.651 \\
 & 5000 & -0.764 & 0.803 & -1.315 & 1.316 & -1.567 & 1.567 \\
\bottomrule
\end{tabular}
\end{table}

In Scenario 3, the primary consequence of naive immigration is a systematic underestimation of the saturation parameter $\mu_1$, the upper boundary of the conditional distribution. The median bias in $\mu_1$ worsens progressively with the proportion of unexposed individuals $\pi$, reaching $-0.64$, $-1.01$, and $-1.23$ at $n = 1000$ for $\pi = 0.10$, $0.25$, and $0.40$ respectively. Although these biases decrease with ample size, they remain substantial at $n = 5000$ ($-0.28$, $-0.59$, $-0.73$), reflecting the fundamental incompatibility between the point mass at $T_i = 0$ for naive individuals and the continuous Beta distribution assumed by the model. A closely related pattern appears in $\phi$, which shows large negative median bias that also worsens with $r$ and persists 
at large $n$ ($-0.76$, $-1.32$, $-1.57$ at $n = 5000$). This is because the excess mass near $T_i = 0$ makes the latent distribution appear more diffuse than the fitted Beta can represent, so the precision parameter is consistently underestimated (we recall that lower values in $\phi$ increse the variance of $T$). By contrast, the age-trend parameters $\eta_0$ and $\eta_1$ exhibit only modest median biases that decrease toward zero as $n$ increases, suggesting that the age-dependent structure is approximately recoverable even in the presence of naive immigration, provided the sample is large enough.

Comparing across the three scenarios, a clear hierarchy of severity emerges. Scenario 2 is the most damaging because it induces structural misspecification in the age-trend parameters $\eta_0$ and $\eta_1$ that persists regardless of sample size, rendering the epidemiological interpretation of the fitted model unreliable even when $n = 5000$. 
Scenario 3 occupies an intermediate position: while the boundary and precision parameters are substantially distorted, the age-trend parameters remain approximately recoverable at large sample sizes, and the direction of the bias in $\mu_1$ provides a diagnostically useful signal of unexposed subpopulation presence. Scenario 1 is the mildest departure, 
with misspecification confined primarily to the lower boundary parameter $\mu_0$ while leaving the age-dependent structure intact. These findings suggest that the proposed modelling framework is most sensitive to 
misspecification that disrupts the global age-trend structure, and most robust to departures that act uniformly across age, which are largely absorbed into the estimated parameters of the conditional distribution of $Y_i$ given $T_i$.

\section{Application to malaria serology}
\label{sec:applications}

We present a re-analysis of malaria serology data from \citet{Bousema2013, Bousema2016}, collected in 2011 from communities in Rachuonyo South District, western Kenyan highlands. Finger-prick blood samples were collected on filter paper and used to detect total immunoglobulin~G (IgG) antibodies against two \textit{Plasmodium falciparum} blood-stage antigens: apical membrane antigen~1 (AMA1) and merozoite surface protein~1 (MSP1). The outcome variable in our analysis is the optical density (OD) measured by ELISA. 

In the analysis, we exclude children under 1 year of age due to the presence of maternally-derived antibodies, which can confound the interpretation of antibody responses acquired through natural infection. However, we do not restrict the data to a specific age range a priori, but instead use our modelling framework to develop a joint model across all ages.

\subsection{Modelling apical membrane antigen~1 (AMA1) concentrations}
\label{sec:ama_analysis}

\begin{figure}[ht!]
    \centering
    \includegraphics[width=1\linewidth]{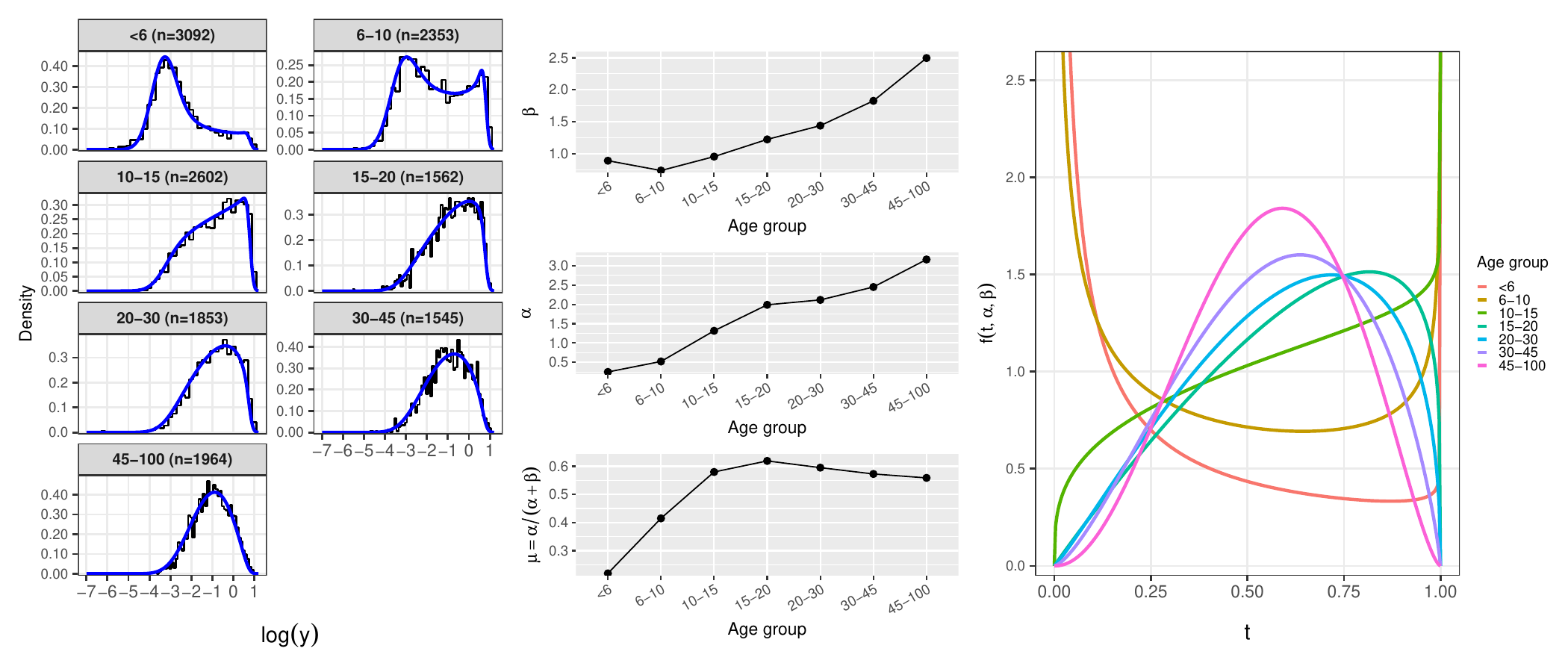}
    \caption{
Left panel: Histograms of the empirical distributions of log AMA1 concentrations within each age group and fitted densities (blue lines). Middle panel: Estimates of the Beta distribution parameters $\alpha$, $\beta$ and its mean $\alpha/(\alpha+\beta)$ across age groups. Right panel: Fitted Beta distributions of the latent variable $T$ for each age group. 
}

    \label{fig:all_ama}
\end{figure}

To gain insight into how to construct a joint model across all ages, we propose the following exploratory approach. We first fit the model in \eqref{eq:lbg}, assuming a Beta distribution with shape parameters $\alpha$ and $\beta$ for the latent state $T$, to the age group 1--6 years. We then refit the same model to successive age groups, namely $[1,6)$, $[6,10)$, $[10,15)$, $[15,20)$, $[20,30)$, $[30,45)$, and $[45,100)$.  For each refit, the parameters of the conditional Gaussian distributions, $\mu_0$, $\mu_1$, $\sigma^2_0$, and $\sigma^2_1$, are fixed to the estimates obtained from the youngest age group ($<6$ years), while only the parameters of the Beta distribution are re-estimated. This model building strategy is based on the assumption that the characteristics of the boundary states, representing individuals with very low seroreactivity ($T$ close to 0) or very high seroreactivity ($T$ close to 1), are biologically constant across the population. What changes with age is the distribution of the latent state $T$, which governs how individuals are distributed between these two extremes. The results of this exploratory step provide empirical guidance for specifying an age-dependent model of the latent state $T$, either through a single Beta distribution with age-varying parameters (Section \ref{sec:age_dependent_shapes}) or through a mixture of Beta distributions allowing for heterogeneous age-related subpopulations (Section \ref{sec:age_mix_prob}).

The results of this step are presented in Figure~\ref{fig:all_ama}. As age increases, the distribution of AMA1 antibody concentrations shifts toward higher values, indicating an overall increase in malaria exposure and infection. Another noteworthy pattern is that the bimodality, or apparent separation between the two mixture components, becomes less distinct above 15 years of age. Examining the estimates of the Beta shape parameters, $\alpha$ and $\beta$, we also observe an overall increase with age, with most of the estimate values for both parameters below 1 (indicating a U-shaped Beta) for ages below 15 years. This supports the initial interpretation based on the histograms: at younger ages, the population is more polarized, whereas with increasing age the fitted Beta distributions transition gradually from a U-shaped to a unimodal form. The increasing values of $\alpha$ and $\beta$ across age groups lead to a mean of $T$ that peaks in the $[15,20)$ interval and then gradually declines, suggesting a mild waning of antibody concentrations at older ages.

Based on the data, individuals below 15 years of age display a distinct immunological profile and antibody dynamics compared with older age groups. To illustrate how this can be accommodated, we adopt a joint modelling framework that accounts for these differences. For younger individuals, we use a mechanistic representation of $T$ based on the reversible catalytic model~\eqref{eq:rev_cat_lambda}, capturing the biological processes of antibody acquisition and boosting. For older individuals, instead, we adopt a more empirical specification that flexibly describes the observed antibody distributions without imposing strong mechanistic assumptions. Furthermore, we treat the age at which this model transition occurs as an additional parameter to estimate. We also point out that this represents one of several possible analyses within our framework, and serves to demonstrate how joint modelling across ages can make fuller use of the available data, avoiding the common practice of dichotomizing or fitting separate models by age group.

To implement this joint model, we specify the latent distribution of $T$ differently below and above an unknown change point $\tau$, while ensuring smooth continuity in the average of $T$ at the transition. For ages below $\tau$, $T$ follows the mechanistic mixture model based on the reversible catalytic formulation in~\eqref{eq:rev_cat_lambda}. For ages above $\tau$, $T$ is described by a single Beta distribution with mean $\mu(a)$ and logit-linear link, as defined in \eqref{eq:logit-log-reg}, and precision parameter $\phi$. The intercept $\eta_0$ is not estimated freely but is determined by the continuity constraint on the mean of $T$ at the change point:
$$
\eta_0 = \mathrm{logit}(\mu_{\tau^-}) - \eta_1 \log(\tau),
$$
where $\mu_{\tau^-}$ denotes the expected mean of $T$ immediately below age $\tau$ implied by the mechanistic model, given by
$$
\mu_{\tau^-} = p_0 \exp(-\tau\lambda)\,\frac{\alpha_1}{\alpha_1 + \beta_1}
           + \bigl(1-p_0 \exp(-\tau\lambda)\bigr)\,\frac{\alpha_2}{\alpha_2 + \beta_2}.
$$
This constraint prevents a discontinuity in the age-dependent mean of $T$ and ensures a biologically coherent transition between the early-life dynamics of antibody acquisition and the slower antibody fluctuations observed in older individuals when the model is fitted jointly. While this continuity constraint may reduce the empirical flexibility of the fit compared to an unconstrained specification, it is imposed to retain a mechanistically interpretable parameterisation that reflects the underlying biological process.

Based on the constraints imposed on the parameters of the Beta mixture model for $T$ (see Section~\ref{sec:age_mix_prob}), an initial fit indicated that $\alpha_1$, $\beta_1$, $\beta_2$, and $p_0$ had estimates converging towards their boundary values. The fitted values were close to 1 for $\alpha_1$, $\beta_2$, and $p_0$, while $\beta_1$ resulted in a very large value of approximately $2990$. These results suggest that the first mixture component can be approximated with a Dirac measure at zero, representing individuals with very low or no seroreactivity. In other words, this finding motivates a simplified formulation in which the first component of $T$ is treated as a degenerate mass at $0$, corresponding to individuals with absent or waning antibody levels, while the second component remains a continuous Beta distribution, with $\alpha_2$ to be estimated and $\beta_2 = 1$, representing individuals with detectable antibody responses arising from exposure to AMA1. The distribution of the latent variable $T$ for individuals below $\tau$ years of age is defined by the mixed discrete-continuous density
\begin{equation}
\label{eq:ama1_model_piecewise}
f_T(t;  a) =
\begin{cases}
1 - \pi(a), & t = 0, \\[6pt]
\pi(a)\,\alpha_2\, t^{\alpha_2 - 1}, & 0 < t < 1, \\[6pt]
0, & \text{otherwise}
\end{cases}
\qquad
.
\end{equation}

\begin{table}[ht!]
\centering
\caption{Maximum likelihood estimates and uncertainty summaries, including standard deviation (SD), the 0.025, 0.5 and 0.975 quantiles, for the joint latent Beta mixture model fitted to the AMA1 data across all ages. Estimates are obtained using a parametric bootstrap based on 1,000 replicates.}
\label{tab:ama_boot}

\begin{tabular}{lrrrrr}
\toprule
\textbf{Parameter} & \textbf{Estimate} & \textbf{SD} & \textbf{2.5\%} & \textbf{50\%} & \textbf{97.5\%}\\
\midrule
\multicolumn{6}{l}{\textit{Distribution of $Y \mid T$}}\\
\addlinespace[2pt]
\quad $\mu_0$ & $-3.194$ & $0.021$ & $-3.237$ & $-3.194$ & $-3.151$\\
\quad $\mu_1$ & $0.747$ & $0.010$ & $0.727$ & $0.747$ & $0.768$\\
\quad $\sigma_0$ & $0.745$ & $0.013$ & $0.719$ & $0.745$ & $0.772$\\
\quad $\sigma_1$ & $0.091$ & $0.013$ & $0.062$ & $0.091$ & $0.117$\\
\addlinespace[4pt]
\quad $\tau$ & $20.842$ & $0.420$ & $20.003$ & $20.876$ & $20.998$\\
\addlinespace[4pt]
\multicolumn{6}{l}{\textit{Distribution of $T$ for age $<\tau$ }}\\
\multicolumn{6}{l}{(with $\alpha_1=1$, $\beta_1=\infty$, $p_0=1$ and $\beta_2=1$)} \\
\addlinespace[2pt]
\quad $\alpha_2$ & $1.498$ & $0.033$ & $1.436$ & $1.499$ & $1.577$\\
\quad $\lambda$ & $0.148$ & $0.005$ & $0.140$ & $0.148$ & $0.158$\\
\addlinespace[4pt]
\multicolumn{6}{l}{\textit{Distribution of $T$ for age $\ge \tau$ }}\\
\addlinespace[2pt]
\quad $\phi$ & $4.544$ & $0.131$ & $4.298$ & $4.551$ & $4.828$\\
\quad $\eta_1$ & $-0.138$ & $0.027$ & $-0.191$ & $-0.135$ & $-0.080$\\
\bottomrule
\end{tabular}
\end{table}

Table~\ref{tab:ama_boot} summarises the parameter estimates and uncertainty from the parametric bootstrap procedure for the joint model fitted across the full age range. The estimated change point for the transition from the mechanistic formulation to the regression-based formulation is $\hat{\tau} = 20.8$ years. This suggests that the mechanistic framework remains appropriate up to approximately age 20. This finding demonstrates that rather than imposing an a priori age restriction, as is sometimes done in serology studies  (e.g. \citet{yman2016}), our model provides an empirical, data-driven approach to determining the appropriate age range for mechanistic modelling. 

The parameters of the mechanistic component for individuals below the change point, $\alpha_2$ and $\lambda$, are consistent with a setting of moderate malaria transmission. We recall that our interpretation of $\lambda$ within our modelling framework is as the rate at which individuals begin to mount measurable antibody responses. The estimated value $\hat{\lambda} = 0.148$ implies that most children develop detectable antibodies within the first few years of life. This means that by age five, approximately $52\%$ of children are expected to have initiated a measurable immune response (computed as $1 - \exp(-0.148 \times 5) \approx 0.52$), with the median age of response around 4.7 years. It is important to notice that, unlike analyses based on reversible catalytic models applied to dichotomised data, this interpretation does not refer to a discrete seroconversion event but rather to the onset of an underlying immunological process that may or may not result in a seropositive test outcome.

For individuals aged above the change point $\tau$, the estimated slope $\eta_1$ in the logit-linear model for $\mu(a)$ is slightly negative, suggesting a modest decline in the mean of $T$ at older ages, potentially reflecting waning antibody levels or reduced boosting due to lower exposure. 

\begin{figure}[ht!]
    \centering
    \includegraphics[width=0.98\linewidth]{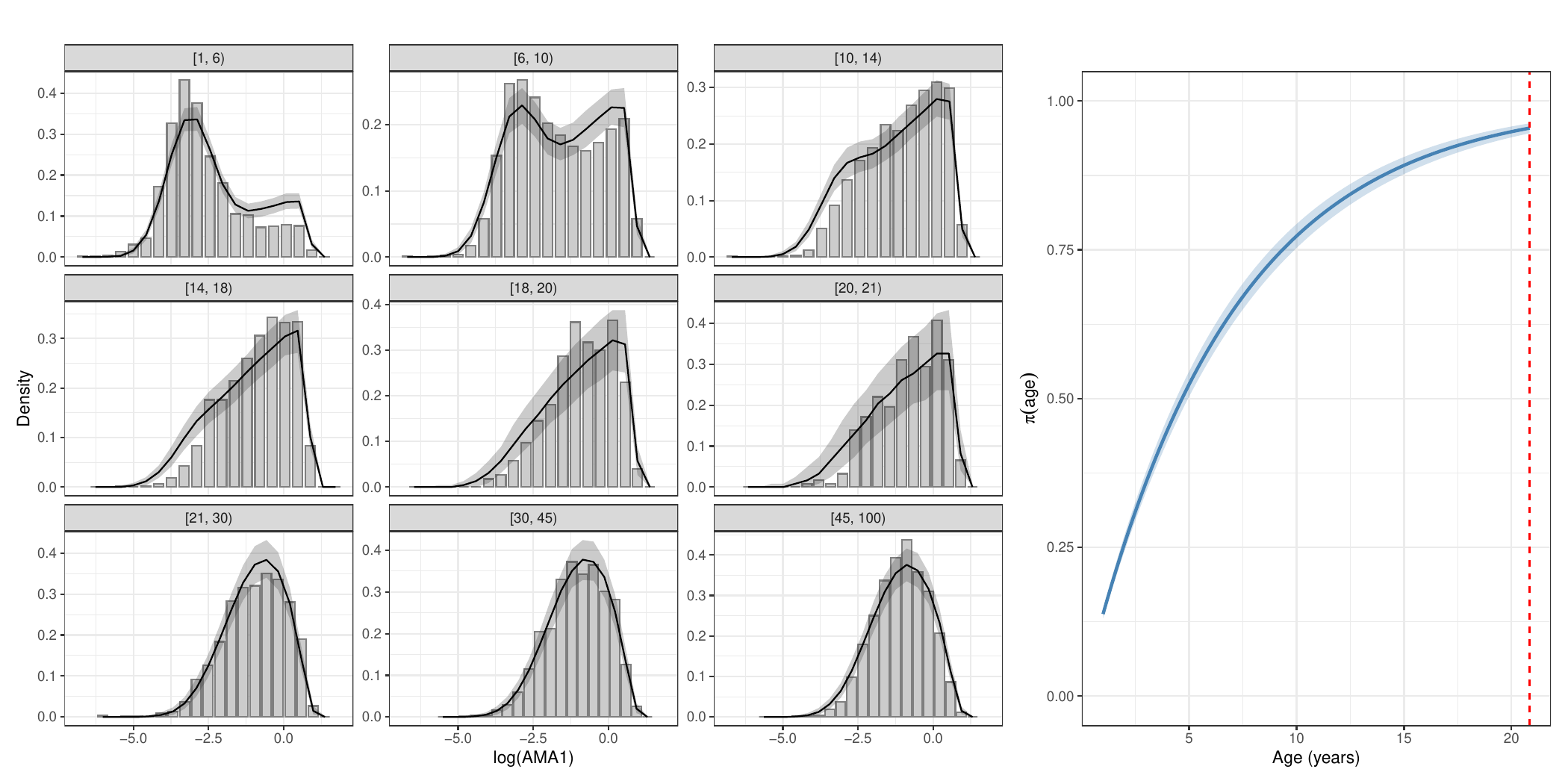}
    \caption{Model validation results for AMA1.
    Left: The plot shows the envelope and the median histogram compared with the empirical distribution of observed AMA1 antibody concentrations (black). 
    Right: Fitted probability of $\pi(a)$ of two-components Beta mixture, as a function of age; the vertical dashed line correspond to the estimated change point parameter $\tau$. \label{fig:ama_plots}
}
\end{figure}

The left panel of Figure~\ref{fig:ama_plots} shows the results of the validation procedure described in Section~\ref{sec:validation}. The main discrepancy between the fitted model and the observed data occurs in the age groups $[1, 6)$ and  $[6,10)$, where the model produces a distribution with a slightly higher concentration of individuals with elevated antibody levels than observed. For the remaining age groups, the observed deviations are minor. Such differences are expected because the model imposes a relatively strong structure on the age-dependent dynamics to ensure the interpretability of its components. As is often the case, a more empirical specification could yield a closer fit to the data, but at the cost of losing biological interpretability.

\subsection{Modelling merozoite surface protein~1 (MSP1) concentrations}

We now turn to MSP1 and apply the same sequence of steps to guide model development under the proposed latent-variable framework. As with AMA1, we begin with the age-stratified exploratory fits described in Section~\ref{sec:age_dependent_shapes}. In these preliminary analyses, the conditional distribution of $Y$ given $T$ is held fixed across age groups, while the Beta distribution governing the latent immune state $T$ is re-estimated within each age band. The resulting fitted densities and age-specific Beta distributions are shown in Figure~\ref{fig:all_msp}.

\begin{figure}[ht!]
    \centering
    \includegraphics[width=1\linewidth]{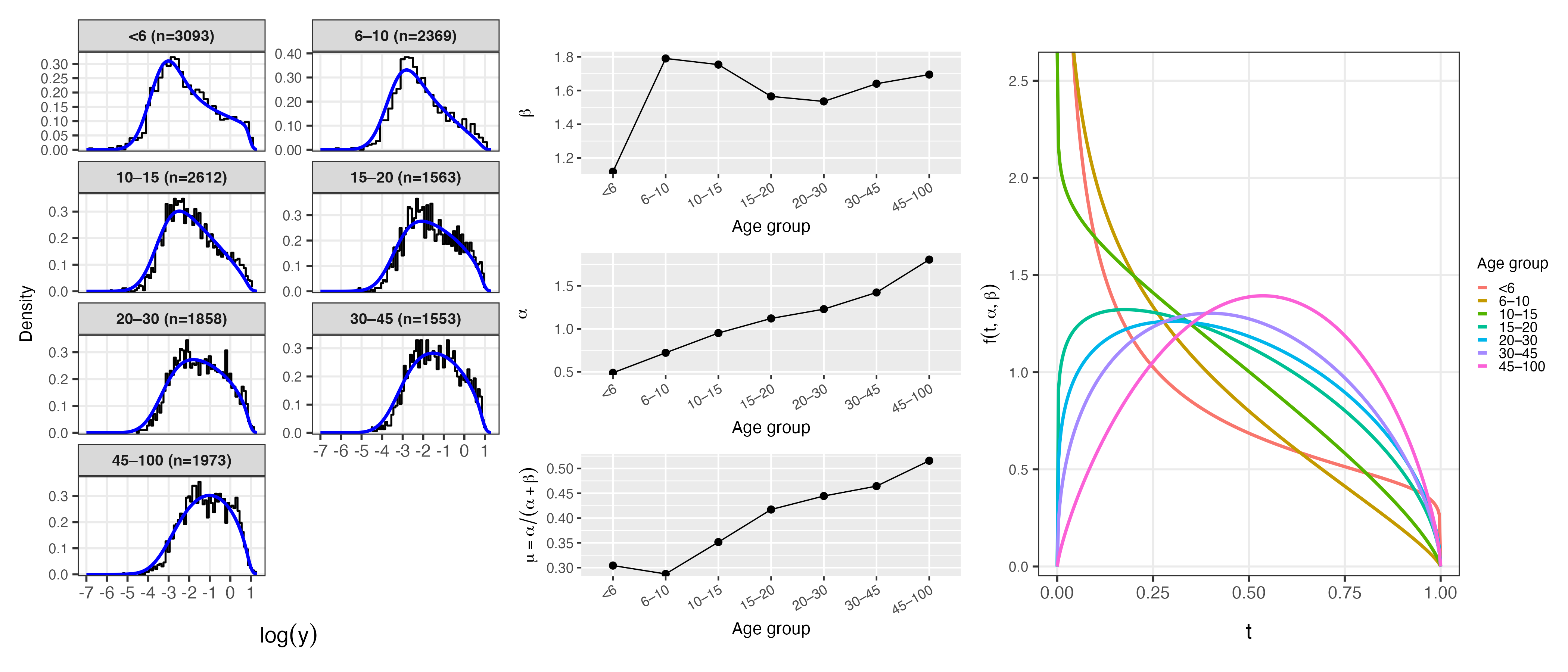}
    \caption{Left panel: Histograms of the empirical distributions of log MSP1 concentrations within each age group and fitted densities (blue lines). Middle panel: Estimates of the Beta distribution parameters $\alpha$, $\beta$ and its mean $\alpha/(\alpha+\beta)$ across age groups. Right panel: Fitted Beta distributions of the latent variable $T$ for each age group. 
}
    \label{fig:all_msp}
\end{figure}

The exploratory fits in Figure~\ref{fig:all_msp} (middle panel) show that, for MSP1, the first shape parameter of the latent distribution, $\alpha$, increases with age, mirroring the pattern observed for AMA1. The parameter $\beta$, instead, stabilises from the $[10,15)$ age group onwards. Consequently, the mean of $T$ rises steadily between the $[10,15)$ and $[45,100)$ age groups. This contrasts with AMA1 (Figure~\ref{fig:all_ama}), where the mean of $T$ begins to decline after the $[15,20)$ age group.

The fitted densities of $Y$ and the corresponding Beta distributions of $T$ (right panel of Figure~\ref{fig:all_msp}) further show that the distinction between individuals with and without detectable antibody responses is strongest in the youngest age group ($[1,6)$ years). At older ages, this separation becomes less distinctive more quickly than for AMA1, as most individuals cluster at higher MSP1 antibody levels. Overall, these patterns align with the expected profile of a long-lasting immune response that accumulates through repeated exposure.

Unlike the analysis for AMA1, here we adopt a more data-driven strategy that fully exploits the flexibility of the latent-variable framework. As we will show, this enables a closer empirical fit to the MSP1 data, where we instead adopted a more mechanistically motivated approach. Based on the insights from Figure~\ref{fig:all_msp}, we formulate an age-dependent model in which $T$ follows a Beta distribution with a smoothly varying $\alpha(a)$ and a piecewise age effect for $\beta(a)$:
$$
T \sim \mathrm{Beta}\big(\alpha(a),\,\beta(a)\big),
$$
with
$$
\alpha(a) = \alpha_0\,a^{\gamma}, 
\qquad
\beta(a) = \beta_0\,a^{\delta(a)},
$$
where
$$
\delta(a) =
\begin{cases}
\delta_1, & a \le \zeta,\\[4pt]
\delta_1+\delta_2, & a > \zeta.
\end{cases},
$$
and $\zeta$ representing the unknown change-point (in years) for the parameter $\delta(a)$. As in the AMA1 analysis, the parameters of the conditional distribution $Y \mid T$ are assumed not to depend on age, so that age-related changes in the observed MSP1 concentrations are driven entirely by the evolution of the latent immune state $T$.

To better understand the properties of the model, let us consider the mean of the latent state $T$ for a given age $a$, after the change point $\zeta$. This becomes
\begin{equation}
\E[T] = \frac{\alpha(a)}{\alpha(a) + \beta(a)}
      = \left(1 + \frac{\beta_0}{\alpha_0}\, a^{\delta_1 + \delta_2 - \gamma}\right)^{-1}, \text{for } a \geq \zeta.
\label{eq:msp_mean}
\end{equation}
Assuming $\gamma > \delta_{1} > 0$, before the change point $\zeta$, both $\alpha(a)$ and $\beta(a)$ increase with age, and the difference $(\gamma - \delta_1)$ determines how quickly the mean in~\eqref{eq:msp_mean} rises from values near 0 to values approaching 1 for $T$. After the change point, the exponent of $\beta(a)$ becomes $\delta_1 + \delta_2$, altering this trajectory. When $\delta_2 < 0$, $\beta(a)$ grows more slowly relative to $\alpha(a)$, which accelerates the increase in $\E[T]$ in \eqref{eq:msp_mean}. This indicates that the population continues to shift towards higher latent immune activation rather than plateauing. A genuine saturation of the immune response would instead require $\delta_1 + \delta_2$ to be close to $\gamma$, so that $\E[T]$ stabilises at around $(1+\beta_0/\alpha_0)^{-1}$. 

\begin{table}[!ht]
\centering
\caption{Maximum likelihood estimates and uncertainty summaries, including standard deviation (SD), the 0.025, 0.5 and 0.975 quantiles, for the latent Beta mixture model fitted to the MSP1 data including all ages. These are obtained using parametric bootstrap based 1,000 replicates.}
\label{tab:msp_boot}

\begin{tabular}{lrrrrr}
\toprule
\textbf{Parameter} & \textbf{Mean} & \textbf{SD} & \textbf{2.5\%} & \textbf{50\%} & \textbf{97.5\%}\\
\midrule
\multicolumn{6}{l}{\textit{Distribution of $Y \mid T$}}\\
\addlinespace[2pt]
\quad $\mu_0$ & $-4.481$ & $0.042$ & $-4.567$ & $-4.479$ & $-4.404$\\
\quad $\mu_1$ & $1.255$ & $0.026$ & $1.205$ & $1.255$ & $1.307$\\
\quad $\log \sigma_0$ & $-0.677$ & $0.043$ & $-0.764$ & $-0.676$ & $-0.599$\\
\quad $\log \sigma_1$ & $-5.716$ & $0.536$ & $-6.892$ & $-5.666$ & $-4.874$\\
\addlinespace[4pt]
\multicolumn{6}{l}{\textit{Distribution of $T$}}\\
\addlinespace[2pt]
\quad $\alpha_0$ & $0.093$ & $0.036$ & $0.025$ & $0.093$ & $0.166$\\
\quad $\gamma$ & $0.277$ & $0.011$ & $0.256$ & $0.277$ & $0.297$\\
\quad $\beta_0$ & $0.755$ & $0.037$ & $0.684$ & $0.754$ & $0.827$\\
\quad $\delta_1$ & $0.110$ & $0.018$ & $0.075$ & $0.110$ & $0.145$\\
\quad $\zeta$ & $11.623$ & $0.406$ & $11.004$ & $11.667$ & $12.197$\\
\quad $\delta_2$ & $-0.061$ & $0.010$ & $-0.080$ & $-0.061$ & $-0.042$\\
\bottomrule
\end{tabular}
\end{table}

\begin{figure}[ht!]
    \centering
    \includegraphics[width=1\linewidth]{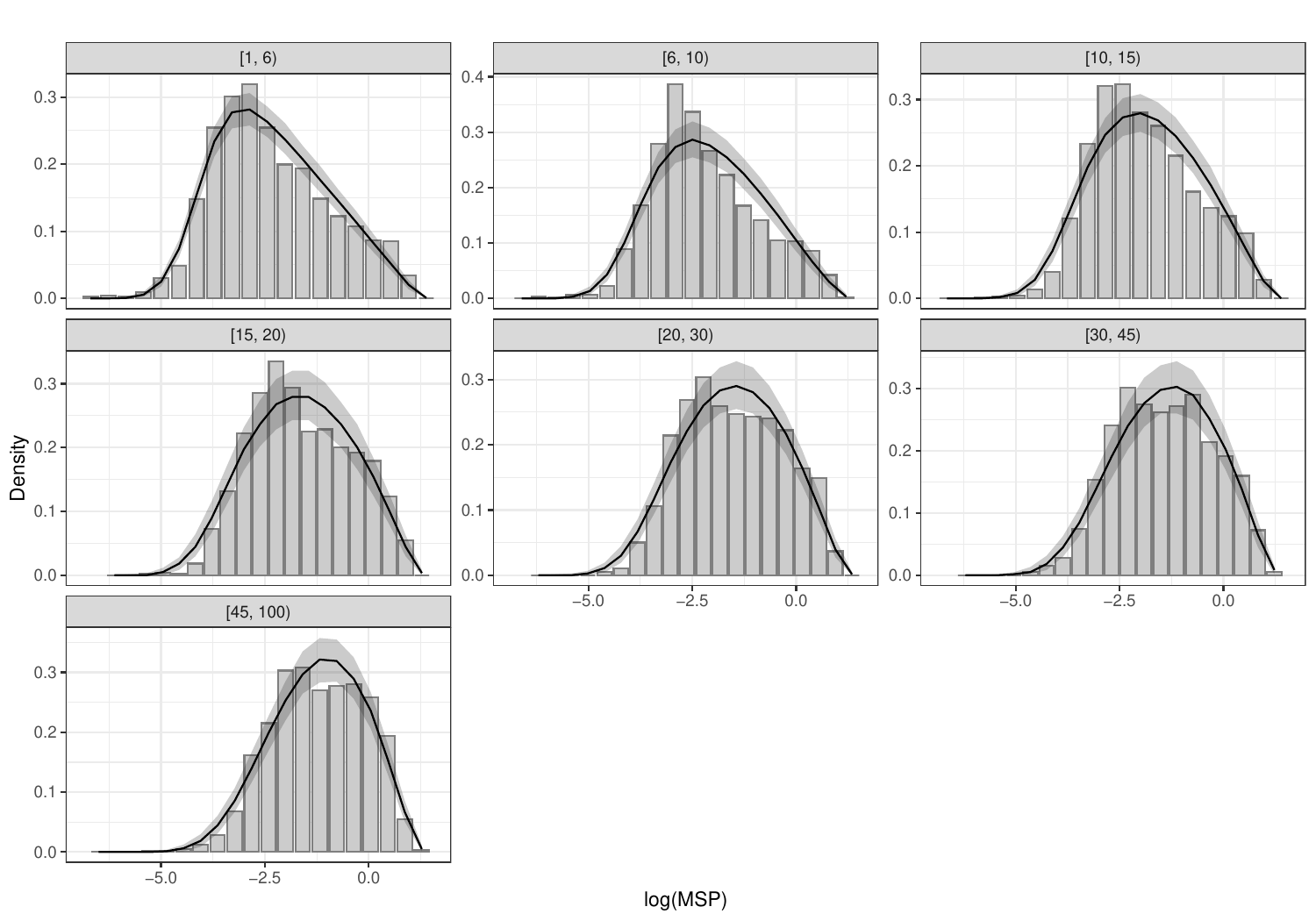}
    \caption{Model validation results for MSP1. The plot shows the envelope (shaded area) and the median histogram (solid black line) compared with the empirical distribution of observed MSP1 antibody concentrations. }
    \label{fig:all_msp_single_beta}
\end{figure}

Table~\ref{tab:msp_boot} summarises the parameter estimates for the MSP1 model, which point to a more durable and persistent antibody response than that observed for AMA1. The estimated exponents $\gamma$ and $\delta_1$ show that both $\alpha(a)$ and $\beta(a)$ increase with age, but that $\alpha(a)$ grows more rapidly, producing a steady rise in the mean latent immune state $\E[T]$ as given by equation~\eqref{eq:msp_mean}. This rise becomes even steeper after the change point at $\zeta = 11.6$ years, as reflected by the negative value of $\delta_2$, which slows the growth of $\beta(a)$ and therefore increases $\E[T]$ more sharply. The validation results in Figure~\ref{fig:all_msp_single_beta} indicate that the model provides a good overall fit, with only minor deviations of the model-based histograms from the empirical histograms observed across the different age groups.

Figure~\ref{fig:expected_lant_by_age} displays the expected mean of $Y$ as a function of age for both MSP1 and AMA1. AMA1 rises sharply in early childhood but then declines after about 20 years, which may reflect the higher individual-level variability in AMA1 antibody dynamics, with rapid boosting and waning following each infection \citep{akpogheneta2008,wipasa2010,ondigo2014}. The peak around 20 years may reflect the higher prevalence of asymptomatic infections with detectable parasitemia in this age range. At this stage, individuals have typically acquired clinical immunity (protection against disease symptoms) but have not yet fully developed anti-parasite immunity, sustaining infections with parasite densities high enough to continue boosting antibody levels. In contrast, as anti-parasite immunity develops in adults, lower parasite densities result in reduced antigenic stimulation and gradual antibody decline.

In contrast, MSP1 shows a steady increase across the full age range with no decline. This is consistent with MSP1 eliciting durable antibody responses \citep{akpogheneta2008,wipasa2010}. The change-point at $\zeta = 11.6$ years marks an acceleration in the rate of increase. While the precise mechanism underlying this shift is unclear, plausible explanations include cumulative exposure effects associated with school attendance, immune maturation during early adolescence, or cohort effects reflecting temporal changes in malaria transmission intensity.
Distinguishing among these possibilities would require longitudinal data or additional historical information on transmission patterns and intervention coverage.

\begin{figure}[!ht]
    \centering
    \includegraphics[width=0.8\textwidth]{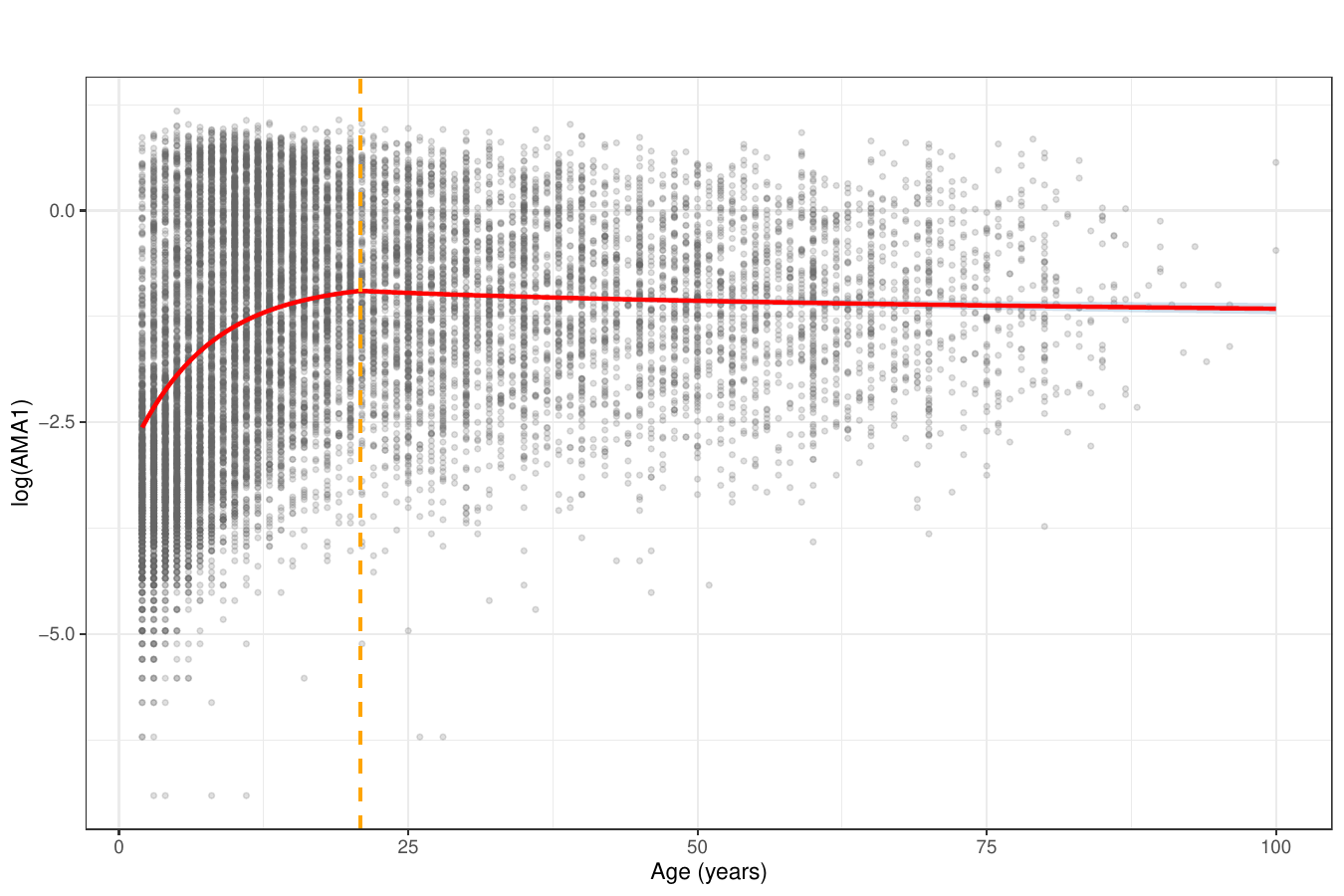}
    \includegraphics[width=0.8\textwidth]{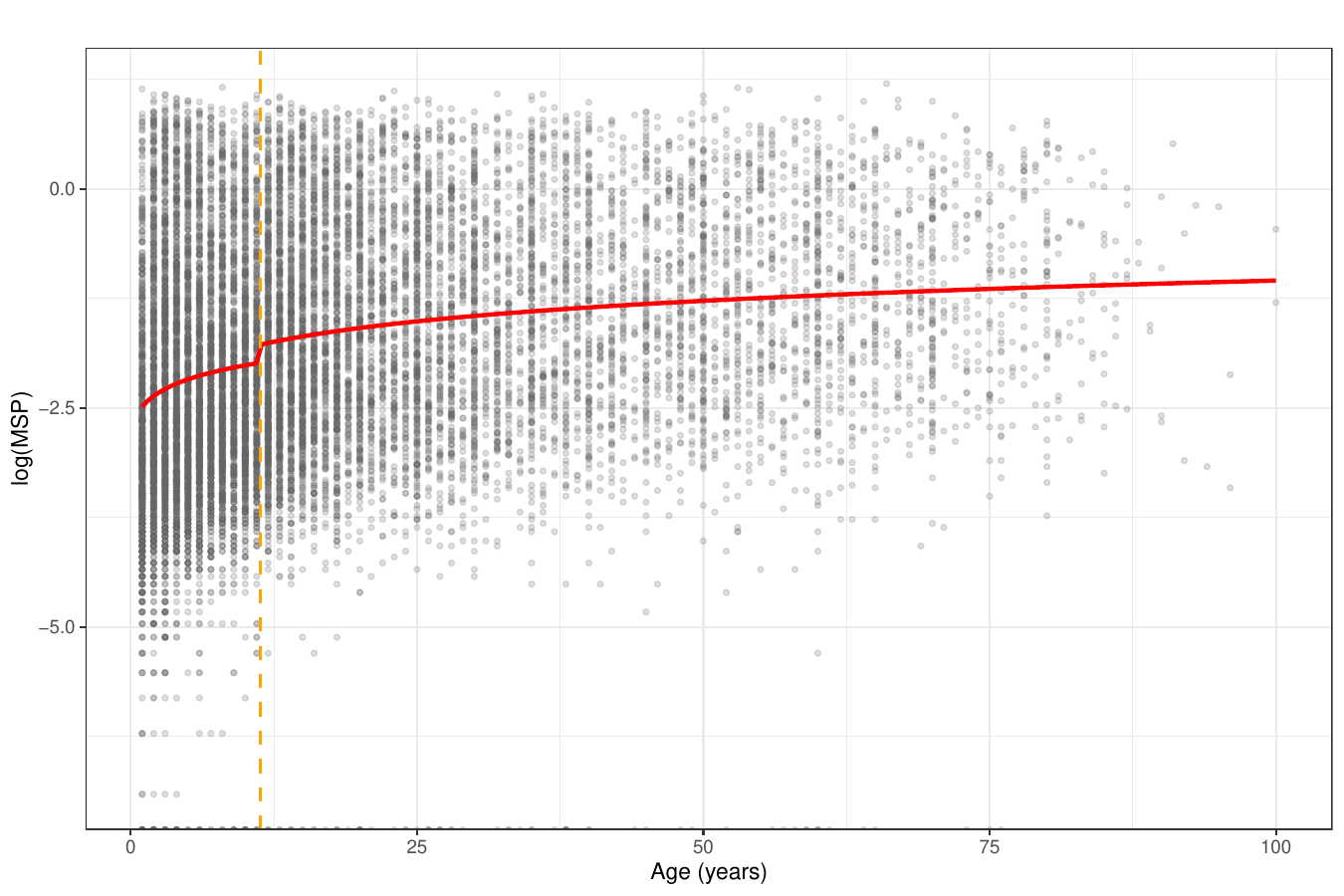}
    \caption{
        Plots of the log antibody levels for AMA1 (left panel) and MSP1 (right panel).
        The solid red curve represents the model-based expected log antibody levels and the black points correspond 
        to the observed log antibody levels.
        The dashed vertical line in the AMA1 panel marks the estimated age threshold (about 20.8 years) 
        where the model specification for the latent variable $T$ changes; see the main text for more details.
    }
    \label{fig:expected_lant_by_age}
\end{figure}

\newpage

\subsection{Model comparison and assessment of computational cost}
\label{subsec:comparison}

\begin{table}[H]
\centering
\caption{Model comparison between the Gaussian mixture model (GMM) and the latent Beta model (LBM). The table reports the Bayesian Information Criterion (BIC) using maximum likelihood estimation (MLE) and $L_2$ histogram-based optimisation, with parameters $\mu_0$, $\mu_1$, $\sigma_0^2$ and $\sigma_1^2$ fixed from the first age group ($<6$ years) using the Gaussian mixture. The BIC difference $\Delta$BIC is defined as $\text{BIC}_{\text{GMM}}-\text{BIC}_{\text{LBM}}$, with positive values indicating a better fit of the LBM.}
\label{tab:model_comparison}
\resizebox{\textwidth}{!}{%
\begin{tabular}{llrrrrrrrr}
\toprule
 &  & \multicolumn{4}{c}{\textbf{MLE}} & \multicolumn{4}{c}{\textbf{$L_2$}} \\
\cmidrule(lr){3-6} \cmidrule(lr){7-10}
\textbf{Antigen} & \textbf{Age} & \textbf{GMM} & \textbf{LBM} & \textbf{$\Delta$BIC} & \textbf{Time} & \textbf{GMM} & \textbf{LBM} & \textbf{$\Delta$BIC} & \textbf{Time} \\
\midrule
AMA1 & $<6$      & 12{,}853 & 12{,}314 & 539  & 3{,}066.450 & 12{,}874 & 13{,}143 & -269 & 8.128 \\
     & $6$--$10$ & 8{,}033  & 7{,}988  & 45   & 70.568      & 8{,}131  & 7{,}994  & 137  & 1.102 \\
     & $10$--$15$& 8{,}565  & 8{,}036  & 529  & 37.449      & 8{,}565  & 8{,}040  & 525  & 0.842 \\
     & $15$--$20$& 4{,}946  & 4{,}470  & 476  & 12.213      & 4{,}946  & 4{,}472  & 474  & 1.366 \\
     & $20$--$30$& 6{,}016  & 5{,}316  & 700  & 13.675      & 6{,}017  & 5{,}318  & 699  & 0.666 \\
     & $30$--$45$& 5{,}084  & 4{,}300  & 784  & 9.989       & 5{,}085  & 4{,}301  & 784  & 1.108 \\
     & $45+$     & 6{,}525  & 5{,}131  & 1{,}394 & 11.018  & 6{,}526  & 5{,}133  & 1{,}393 & 0.857 \\
\addlinespace[4pt]
MSP1 & $<6$      & 13{,}277 & 13{,}235 & 42   & 1{,}181.518 & 13{,}328 & 13{,}291 & 37   & 7.540 \\
     & $6$--$10$ & 7{,}816  & 7{,}654  & 162  & 61.514      & 7{,}816  & 7{,}674  & 142  & 0.750 \\
     & $10$--$15$& 8{,}678  & 8{,}364  & 314  & 33.332      & 8{,}679  & 8{,}382  & 297  & 1.617 \\
     & $15$--$20$& 5{,}184  & 4{,}973  & 211  & 17.352      & 5{,}185  & 4{,}980  & 205  & 1.255 \\
     & $20$--$30$& 6{,}129  & 5{,}870  & 259  & 19.832      & 6{,}129  & 5{,}878  & 251  & 1.238 \\
     & $30$--$45$& 5{,}071  & 4{,}843  & 228  & 13.780      & 5{,}071  & 4{,}849  & 222  & 1.169 \\
     & $45+$     & 6{,}391  & 5{,}932  & 459  & 14.386      & 6{,}391  & 5{,}941  & 450  & 1.244 \\
\bottomrule
\end{tabular}
}
\end{table}

We conduct an age-stratified analysis to compare the proposed latent Beta model (LBM) with the standard Gaussian mixture model (GMM). For the LBM, we assume a single Beta distribution for the latent variable $T$, with shape parameters $\alpha$ and $\beta$ to be estimated. For both models, the conditional distribution parameters $\mu_0$, $\mu_1$, $\sigma_0$, and $\sigma_1$ are fixed at values estimated by fitting the GMM and LBM, respectively, to the youngest age group ($<6$ years), which shows the clearest separation between individuals with low and high antibody levels. For each subsequent age group, we re-estimate only the mixing proportion $\pi$ in the GMM or the Beta shape parameters $\alpha$ and $\beta$ in the LBM. Since the GMM can be viewed as a special case of the LBM in which the latent distribution collapses to a discrete two-point distribution, this constrained comparison directly assesses whether age-related variation in antibody responses can be explained through changes in the latent distribution structure for the LBM, as opposed to changes in mixture weights for the GMM. 

The results of this comparison are presented in Table~\ref{tab:model_comparison}. For LBM the fitted densities using MLE are shown in the left panel of Figure \ref{fig:all_ama} and Figure \ref{fig:all_msp}; for GMM these are shown in Figure \ref{fig:compare_gmm_ama} and Figure \ref{fig:compare_gmm_msp} of the Appendix. Under maximum likelihood estimation, the LBM consistently outperforms the GMM across all age groups and both antigens. For AMA1, the BIC advantage of the LBM increases with age and becomes substantial from age 10 years onwards, often exceeding 100 BIC units in difference. For MSP1, the BIC differences in favour of the LBM are generally smaller than those observed for AMA1, indicating that the GMM is better able to adjust to age-related changes in MSP1 antibody distributions than for AMA1. Nevertheless, the LBM still provides a superior fit across all age groups, suggesting that a continuous latent representation remains more appropriate even when the Gaussian mixture adapts relatively well.

When using the $L_2$ histogram-based optimisation, the qualitative pattern of results remains largely unchanged. The only exception is the youngest age group for AMA1, where the $L_2$ criterion slightly favours the GMM, reflecting the fact that both models rely on parameters fixed using this age group and that the $L_2$ approximation places greater emphasis on matching the empirical histogram. Aside from this case, the BIC differences obtained under $L_2$ optimisation are highly consistent with those derived from maximum likelihood, both in direction and magnitude.

From a computational perspective, the $L_2$ histogram-based optimisation yields substantial efficiency gains relative to maximum likelihood estimation. For the LBM, speedups range from more than an order of magnitude for moderate-sized age groups to several hundred-fold for the largest groups, reducing computation times from tens of minutes to a few seconds. These gains are particularly pronounced for AMA1, but are also substantial for MSP1. Importantly, this dramatic reduction in computational cost is achieved without altering our conclusions, indicating that the $L_2$ approach provides a reliable and efficient approximation for large-scale model comparisons.

\section{Discussion}
\label{sec:discussion}

We have introduced a latent variable modelling framework that offers four key methodological contributions. First, it provides a principled way to model the full distribution of antibody concentrations without dichotomization, which remains the dominant practice in serological analyses \citep{corran2007,Drakeley2005,Arnold2014,cox2022} despite well-established information loss \citep{royston2006,kyomuhangi2021}. By modelling antibody levels directly, the framework provides a more refined characterisation of how immune responses vary and evolve with age. Second, it offers a flexible approach for incorporating age-dependent structure through either mechanistic specifications that encode epidemiological assumptions about exposure and immune response dynamics, or data-driven specifications that adaptively capture empirical patterns without imposing strong structural constraints. Third, by explicitly representing immune activation as a continuum, the framework more naturally accommodates complex population heterogeneity in immunological responses. Fourth, we introduce an $L_2$ histogram-based estimator that provides a computationally efficient alternative to maximum likelihood estimation, enabling practical implementation of the framework even in computationally demanding contexts with large sample sizes.

A further distinction must be made between mechanistic models applied to the latent variable $T$ in our framework and the same mechanistic models traditionally applied to dichotomised antibody outcomes. For example, when the reversible catalytic model in \eqref{eq:rev_cat_lambda} is used for binary serostatus data, the parameter $\lambda$ is interpreted as the seroconversion rate, that is, the rate at which individuals transition from seronegative to seropositive status. In our framework, however, seroconversion is not a binary event, and $T$ represents a continuous underlying immune state. Consequently, the parameter $\lambda$ should be interpreted as the rate at which individuals develop stronger seroreactivity, which may or may not lead to a seropositive test outcome. This thus suggests that estimates of $\lambda$ obtained under our model are expected to be higher than those derived from dichotomised serological analyses. Another important feature of the framework is that it enables a more rigorous statistical validation of the biological assumptions underlying mechanistic models, an aspect that is rarely examined in standard serological analyses. For instance, the application in Section~\ref{sec:applications} reveals subtle deviations from the reversible catalytic model in the immune response dynamics of children aged below 10 years. Such discrepancies would remain undetected in analyses based on dichotomized serostatus data.

When incorporating age into the distribution of the latent variable $T$, the choice between a single Beta distribution and a finite mixture of Beta distributions should be driven primarily by the inferential objectives of the analysis rather than by the need for increased distributional flexibility \textit{per se}. Although a mixture offers greater flexibility, the trade-off with identifiability remains a practical concern. More fundamentally, the two approaches may yield near-identical fits while carrying qualitatively different interpretations. For example, a U-shaped distribution on $(0,1)$ can be expressed equivalently using a single Beta with $\alpha < 1, \beta < 1$ or as a two-component mixture concentrated near the boundaries; however, only the latter allows explicit modelling of subpopulations with distinct seroreactivity profiles whose age-dependent prevalence can be linked to mechanistic formulations such as the catalytic model in \eqref{eq:rev_cat_lambda}. The single Beta distribution might be a more natural choice when the goal is a descriptive characterization of how age shapes the distribution of $T$ and, hence, the observed antibody levels. The mixture formulation instead becomes advantageous when the scientific question explicitly concerns disentangling immunologically distinct subgroups within the population, and is most
naturally applied in contexts where the age dynamics are sufficiently well understood to support a mechanistic parameterization of the mixing probabilities, as in the malaria case studies considered here.

When using a mixture distribution to define the distribution of the latent variable $T$, an open question remains on how many components one should use. In particular, whether this number should be fixed \emph{a priori} or determined from the data
depends fundamentally on the inferential goal. In applications directed at explaining mechanisms of antibody acquisition, as in Section~\ref{sec:age_mix_prob}, the number of components is a deliberate scientific choice: the two-component specification encodes a biologically interpretable contrast between individuals with low and high immune activation, and its justification comes from subject-matter knowledge rather than data-adaptive criteria. However, when the framework is extended to the joint modelling of responses to
multiple antigens---on which we provide more discussion below---the goal may lie in the identification of latent immunological subgroups whose number is not known
in advance. In this setting, modelling the multivariate latent state through a mixture distribution offers a conceptually appealing alternative to standard clustering methods, since the clustering is performed in a continuous latent space rather than directly on noisy antibody measurements. For such applications, data-driven selection of the number of components is both meaningful and, we believe, practically feasible. Future research could thus extend the current framework to multivariate latent clustering by adapting existing approaches for selecting the number of mixture components, both in a Bayesian and frequentist context (see for example \citealt{richardsongreen1997, rousseau2011,
malsiner2016}).

A central assumption of our modelling framework is that the conditional distribution parameters ($\mu_0$, $\mu_1$, $\sigma_0^2$, $\sigma_1^2$) reflect primarily assay characteristics rather than biological variation, with $T$ capturing the underlying immune state. This assumption is most credible within a single laboratory using consistent protocols, but becomes problematic when comparing across settings where boundary parameters may be confounded with biological differences in exposure histories or genetic backgrounds. For antigens eliciting qualitatively different immune responses, the saturation parameter $\mu_1$ may reflect both the assay's upper detection limit and genuine biological constraints on antibody production or persistence. For these reasons, the proposed latent based modelling approach does not automatically achieve standardisation via $T$ across laboratories or assay platforms. For $T$ to serve as a comparable metric across settings, calibration using shared reference standards would be required to anchor the boundary parameters on a common scale. Without such calibration, $T$ estimates should be interpreted as relative measures within each specific assay context rather than absolute, universally comparable quantities. Regarding assay saturation, while the framework explicitly models $\mu_1$ as a measurement ceiling, it cannot distinguish individuals beyond this limit when both produce saturated measurements. The unit interval $[0,1]$ for $T$ is a modelling convenience; the endpoints represent extremes of observable variation rather than absolute biological states. In high transmission settings with widespread saturation, the upper range of $T$ may become poorly identified. 

An important consideration for applying our framework concerns vaccination status. For diseases conferring durable immunity following vaccination or natural infection --- e.g.\ yellow fever, where a single vaccine dose provides lifelong protection \citep{Staples2015} --- binary serological classification may suffice for most purposes, and continuous antibody modelling may offer limited added value. However, for diseases characterised by partial or waning immunity, such as malaria, continuous antibody measurements provide information about immune activation dynamics that is lost when data are dichotomised. The recent introduction of malaria vaccines (RTS,S/AS01 and R21/Matrix-M) \citep{RTS2015,Datoo2021} presents new challenges for serological surveillance that affect any analytical approach, not just our modelling framework. Because these vaccines target the pre-erythrocytic circumsporozoite protein rather than the blood-stage antigens we analyse (AMA1, MSP1), vaccination does not directly boost these antibody responses. However, reduced transmission resulting from vaccination programmes could alter natural boosting patterns, thereby affecting the age-dependent dynamics we model. Our simulation study provides evidence on the potential consequences of such departures. When interventions uniformly elevate antibody levels across all ages, the estimated age-dependent structure remains largely intact, with misspecification absorbed primarily into the parameters governing the conditional distribution of antibody levels given the latent seroreactivty state. By contrast, when interventions selectively reduce exposure within specific age groups, inference on the age-dependent profile of seroreactivity becomes unreliable, with biases in the estimated age structure that persist regardless of sample size. In such settings, vaccination status should be incorporated into the model either as a covariate or through a stochastic process affecting the latent immune state $T$, depending on the research question and disease context. When data on multiple antigens are available, joint modelling of their responses might further help to distinguish vaccine-induced from infection-induced seroreactivity.

Our simulation study demonstrates that the $L_2$ estimator achieves speedup factors ranging from 7-fold at $n=100$ to over 250-fold at $n=5000$ relative to MLE, while maintaining reasonable statistical accuracy in most scenarios. However, MLE demonstrates superior asymptotic efficiency, particularly in challenging scenarios with symmetric or highly skewed latent distributions, where $L_2$ exhibits persistent bias even at large sample sizes. We therefore recommend a hybrid strategy for parameter estimation which consists of using the $L_2$ estimation for exploratory analysis and, whenever computationally feasible, refine final estimates using MLE with $L_2$ solutions as starting values. However in scenarios where there is a clear bimodal separation in the distribution of antibody levels, $L_2$ performs nearly as well as MLE even at moderate sample sizes and might be used without the need for further refinement of the estimates with MLE.

Future research will  focus on extending the current modelling framework to address a wider range of questions in sero epidemiological research, including the development of joint models for multiple antibody responses. One natural direction is to replace the single latent quantity with a multivariate latent structure and to combine this with an appropriate multivariate Beta distribution, for example using the construction proposed by \citet{arnold2011}. This would allow the model to capture dependencies between antibody responses to different antigens and provide deeper insight into how immune markers coevolve following exposure. The same idea extends naturally beyond serology, since the latent variable formulation is directly relevant to high dimensional omics settings such as proteomics and transcriptomics, where disease states are typically treated as discrete categories, for example healthy versus diseased or sepsis molecular endotypes. Recent work using neural network models for disease classification from multi omics data \citep{hartman2023} and large transcriptomic and proteomic studies of sepsis subtypes \citep{wong2012, reyes2020} all adopt this discrete label perspective, despite the underlying host response being biologically continuous rather than strictly partitioned. Within our framework this continuum can be represented by a latent variable $T$, with each biomarker, here each transcript or protein, modelled as an observed outcome $Y_{j}$ whose conditional distribution depends smoothly on $T$ through marker specific mean and variance functions that describe how expression evolves across the response axis. While a single latent variable provides a parsimonious representation of the dominant gradient of biological variation, many omics settings involve several distinct processes, for example inflammation, immune suppression, and endothelial injury, which may jointly determine the multivariate profile. In such cases it may be necessary to adopt a multivariate latent representation and to define the conditional distribution of $Y_{j}$ as a function of this collection of latent quantities, allowing the framework to capture more complex patterns and multiple dimensions of biological variation that cannot be reduced to a single continuum.

Another important avenue concerns the inclusion of spatial and temporal structure in the latent process $T$. Several options are available, as this can be achieved by introducing correlated random effects into the parameters of the Beta distribution for $T$. For example, using the mean-variance parameterisation described in Section~\ref{sec:age_dependent_shapes}, spatially or temporally correlated random effects could be introduced into the mean of $T$. This would support applications in spatial mapping, where the aim is to identify geographic patterns in immunity and detect potential immunity gaps, as well as in longitudinal studies, where interest lies in modelling serial correlation and characterising how an individual's immune response evolves over time, for example after an intervention. 

In summary, the flexibility of the proposed latent variable modelling framework, together with its strong potential for further extensions, illustrate its broad applicability in sero epidemiological research and beyond. We believe that these strengths position it as a compelling approach for routine use in settings where a deeper understanding of immunological processes is of primary importance.

\section*{Acknowledgements}

We thank all those who contributed to the collection of data included in this paper, specifically the survey participants in Kenya, and the KEMRI/CDC research team.

We would like to thank Prof. Peter J. Diggle (Lancaster University), Dr. Irene Kyomuhangi (Lancaster University), Ivan Hejný (Lund University), Erik Hartman (Lund University), Prof. Joacim Rocklöv (University of Heidelberg), and Dr. Gillian Stresman (University of South Florida) for their valuable feedback, which helped to improve the manuscript.

The computations described in this paper were performed using the University of Birmingham's BlueBEAR HPC service, which provides a High Performance Computing service to the University's research community. See \url{www.birmingham.ac.uk/bear} for more details.

\bibliographystyle{plainnat}
\bibliography{references}

\appendix

\section{Proof of Theorem \ref{lem:L_2-consistency}}
\label{sec:proof-theorem}

\begin{proof}
Let $f_{\vartheta, n} :=\; f_{n}(\cdot ; \vartheta)$. We apply Corollary 3.2.3 of \cite{vandervaartwellner1996}. This corollary establishes consistency provided that: (i) the empirical criterion converges uniformly to the population criterion, i.e., $\sup_{\vartheta \in \vartheta} |\mathbb{M}_n(\vartheta) - \mathbb{M}(\vartheta)| \xrightarrow{p} 0$.  and  that (ii) the population maximum is well-separated. Assumption \ref{asm:ident} directly gives condition (ii). Thus, it remains only to show condition (i).

We prove this condition in two step.
We first establish that the unparameterized histogram $\widehat{f}_n$ converges to the true density $f_0$ in the $L_2$ norm. \cite{freedman1981histogram} provide the exact mean integrated squared error (MISE) decomposition in their Proposition (1.10):
\begin{equation*}
    \mathbb{E}\Bigl[ \|\widehat{f}_n - f_0\|_2^2 \Bigr] \;=\; \left( \frac{1}{n h_n} - \frac{1}{n} \|f_{0,n}\|_2^2 \right) \;+\; \|f_{0,n} - f_0\|_2^2,
\end{equation*}
where $f_{0,n}$ is the theoretical discretized density (the expectation of $\widehat{f}_n$). 
The first term is bounded by $(n h_n)^{-1}$ and under Assumption \ref{asm:histogram}, this term vanishes.
The second term corresponds to the approximation error of the binned density; as shown in Eq.~(2.5) of \cite{freedman1981histogram}, this vanishes as $h_n \to 0$ for any $f_0 \in L_2$.
Therefore, $\mathbb{E}[\|\widehat{f}_n - f_0\|_2^2] \to 0$, and by Markov's inequality:
\begin{equation}
\label{eq:hist-conv}
    \|\widehat{f}_n - f_0\|_2 \;\xrightarrow{p}\; 0.
\end{equation}

Using the first step we establish Uniform convergence of the criterion.
Note that $|\mathbb{M}_n(\vartheta) - \mathbb{M}(\vartheta)| = | \|\widehat{f}_n - f_{\vartheta,n}\|_2^2 - \|f_0 - f(\cdot;\vartheta)\|_2^2 |$.
We decompose the difference as:
\begin{align*}
    |\mathbb{M}_n(\vartheta) - \mathbb{M}(\vartheta)| \;&\le\; \Bigl| \|\widehat{f}_n - f_{\vartheta,n}\|_2^2 - \|f_0 - f_{\vartheta,n}\|_2^2 \Bigr|  \\
    \;&\quad +\; \Bigl| \|f_0 - f_{\vartheta,n}\|_2^2 - \|f_0 - f(\cdot;\vartheta)\|_2^2 \Bigr|.
\end{align*}
We prove separately that the two terms on the right goes to zero.
By the triangle inequality for norms, $| \|a\|_2^2 - \|b\|_2^2 | \le \|a - b\|_2 (\|a\|_2 + \|b\|_2)$. Applying the inequality to the first term (with $a = \widehat{f}_n - f_{\vartheta,n}$ and $b = f_0 - f_{\vartheta,n}$):
\begin{equation*}
    \Bigl| \|\widehat{f}_n - f_{\vartheta,n}\|_2^2 - \|f_0 - f_{\vartheta,n}\|_2^2 \Bigr| \;\le\; \|\widehat{f}_n - f_0\|_2 \left( \|\widehat{f}_n - f_{\vartheta,n}\|_2 + \|f_0 - f_{\vartheta,n}\|_2 \right).
\end{equation*}
 By Assumption \ref{asm:regularity}, the terms in the parenthesis are bounded uniformly in probability. Thus by \eqref{eq:hist-conv}, $ \Bigl| \|\widehat{f}_n - f_{\vartheta,n}\|_2^2 - \|f_0 - f_{\vartheta,n}\|_2^2 \Bigr| $ converges to 0 uniformly in $\vartheta$.

Applying the inequality to the second term:
\begin{equation*}
   \Bigl| \|f_0 - f_{\vartheta,n}\|_2^2 - \|f_0 - f(\cdot;\vartheta)\|_2^2 \Bigr|\;\le\; \|f_{\vartheta,n} - f(\cdot;\vartheta)\|_2 \left( \|f_0 - f_{\vartheta,n}\|_2 + \|f_0 - f(\cdot;\vartheta)\|_2 \right).
\end{equation*}
Again, following Assumption \ref{asm:regularity} and  \eqref{eq:hist-conv}, the terms in parentheses are bounded uniformly and bounded in probability (note that if two terms tend to zero in probability, so does their product). \cite{freedman1981histogram} show in Proposition (2.7) that when the density is absolutely continuous with square-integrable derivative, then $\|f_{\vartheta,n} - f(\cdot;\vartheta)\|^2_2 $  is $O(h_n^2)$. Specifically, for small $h_n$, the term is bounded by a constant times $h_n^2 \|f'(\cdot;\vartheta)\|_2^2$.
Thus under Assumption \ref{asm:smoothness}, $\sup_{\vartheta \in \vartheta} \|f'(\cdot;\vartheta)\|_2 < \infty$:
\begin{equation*}
    \sup_{\vartheta \in \vartheta} \|f_{\vartheta,n} - f(\cdot;\vartheta)\|_2 \;\le\; C h_n \sup_{\vartheta \in \vartheta} \|f'(\cdot;\vartheta)\|_2 \;\longrightarrow\; 0.
\end{equation*}
Thus, $\sup_{\vartheta \in \vartheta} |\mathbb{M}_n(\vartheta) - \mathbb{M}(\vartheta)| \xrightarrow{p} 0$.

\end{proof}

\section{Verification of regularity conditions for the latent variable model}
\label{subsec:regularity-verification}

To apply the consistency result of Theorem \ref{lem:L_2-consistency} to our proposed latent variable model, we must verify that the model defined in \eqref{eq:lbg} satisfies Assumptions \textbf{A1}--\textbf{A5}.

\begin{lemma}
\label{lem:mixture-regularity}
Consider the latent variable model 
\begin{equation*}
    f(y;\vartheta) \;=\; \int_0^1 \phi(y; \mu(t), \sigma^2(t)) \, g_T(t; \psi) \, dt,
\end{equation*}
where $\mu(t) = (1-t)\mu_0 + t\mu_1$ and $\sigma^2(t) = (1-t)\sigma_0^2 + t\sigma_1^2$. Assume the parameter space $\vartheta$ is compact (\textbf{A1}) and satisfies the ordering constraint $\mu_0 < \mu_1$ to avoid label switching. Furthermore, assume $\sigma_0, \sigma_1 \ge c$ for some constant $c > 0$ to ensure the density is bounded. Let $g_T$ be identifiable and continous with respect  $\psi$. a If the histogram discretization satisfies \textbf{A4}, then the estimator satisfies Assumptions \textbf{A1}--\textbf{A5}.
\end{lemma}

\begin{proof}
We verify each assumption separately:
\begin{enumerate}[label=\textbf{A\arabic*.}, leftmargin=3em, itemsep=0.5ex]
    \item \textbf{Compactness.} Satisfied by the definition of the parameter space $\vartheta$.

\item \textbf{Identification.} 
The condition requiring a well-separated maximum is satisfied by the combination of model identifiability, continuity, and compactness.

First, the mapping $\vartheta \mapsto f(\cdot; \vartheta)$ is injective, ensuring that the population criterion $\mathbb{M}(\vartheta)$ has a unique maximizer $\vartheta_0$. This injectivity holds because the $g_T(t; \psi)$ is identifiable, and the ordering constraint $\mu_0 < \mu_1$ ensures a strictly monotonic (bijective) mapping from the latent variable to the conditional moments.

Second, the criterion $\mathbb{M}(\vartheta)$ is continuous on $\vartheta$ because both the Gaussian kernel $\phi$ and the mixing density $g_T$ are continuous functions of the parameters. To see that the maximum is well-separated, consider any open neighborhood $G$ of $\vartheta_0$. The set $K = \vartheta \setminus G$ is a closed subset of a compact space, and thus is itself compact. By the Extreme Value Theorem, the continuous function $\mathbb{M}$ must attain a maximum on $K$. Since $\vartheta_0$ is the \emph{unique} maximizer and $\vartheta_0 \notin K$, this attained maximum must be strictly less than $\mathbb{M}(\vartheta_0)$:
\begin{equation*}
    \sup_{\vartheta \in \vartheta \setminus G} \mathbb{M}(\vartheta) \;=\; \max_{\vartheta \in K} \mathbb{M}(\vartheta) \;<\; \mathbb{M}(\vartheta_0).
\end{equation*}
    
    \item \textbf{Boundedness.} We must show $f(y;\vartheta)$ is uniformly bounded. The Gaussian kernel $\phi(y; \mu, \sigma^2)$ is bounded by $(2\pi\sigma^2)^{-1/2}$. By the assumption that component variances are bounded away from zero ($\sigma_0, \sigma_1 \ge c$), the interpolated variance satisfies $\sigma(t) \ge c$ for all $t \in [0,1]$. Thus,
    $$
    f(y;\vartheta) \;=\; \int_0^1 \phi(y; \mu(t), \sigma^2(t)) g_T(t; \psi) \, dt \;\le\; \frac{1}{\sqrt{2\pi}c} \int_0^1 g_T(t; \psi) \, dt \;=\; \frac{1}{\sqrt{2\pi}c}.
    $$
    Since $f$ is bounded by a constant $K$ and $\int f(y) dy = 1$, it follows that $\|f\|_2^2 \le K \int f(y) dy = K < \infty$, so $f \in L_2$.

    \item \textbf{Histogram Regime.} This property can be easily satisfied by choosing any histogram binning rule for which $h_n \to 0$ and $n h_n \to \infty$, such as the Sturges rule adopted in this study with $h_n = (y_{(n)}-y_{(1)})/(1+\log_2 n)$.

    \item \textbf{Smoothness.} We require the derivative $f'(y;\vartheta)$ to be in $L_2$. Because $\sigma(t) \ge c > 0$, the Gaussian kernel is smooth with bounded derivatives. We differentiate under the integral sign:
    $$
    f'(y;\vartheta) \;=\; \int_0^1 \frac{\partial}{\partial y} \phi(y; \mu(t), \sigma^2(t)) \, g_T(t; \psi) \, dt.
    $$
    Recall that $\frac{\partial}{\partial y} \phi(y; \mu, \sigma^2) = -\frac{y-\mu}{\sigma^2}\phi(y)$.
    The $L_2$ norm of this derivative for a single Gaussian component is proportional to $\sigma^{-3/2}$. Since $\sigma(t) \ge c$, there exists a uniform bound $M < \infty$ such that $\|\phi'(\cdot|t)\|_2 \le M$ for all $t \in [0,1]$. By Minkowski's inequality for integrals:
    $$
    \|f'(\cdot;\vartheta)\|_2 \;\le\; \int_0^1 \|\phi'(\cdot|t) g_T(t; \psi)\|_2 \, dt \;=\; \int_0^1 g_T(t; \psi) \|\phi'(\cdot|t)\|_2 \, dt.
    $$
    Using the uniform bound $M$:
    $$
    \|f'(\cdot;\vartheta)\|_2 \;\le\; \int_0^1 g_T(t; \psi) M \, dt \;=\; M.
    $$
    Thus, $\sup_{\vartheta \in \vartheta} \|f'(\cdot;\vartheta)\|_2 < \infty$, satisfying the smoothness condition.
\end{enumerate}
\end{proof}

\section{Additional plots}

\begin{figure}[H]
    \centering
    \includegraphics[width=0.8\linewidth]{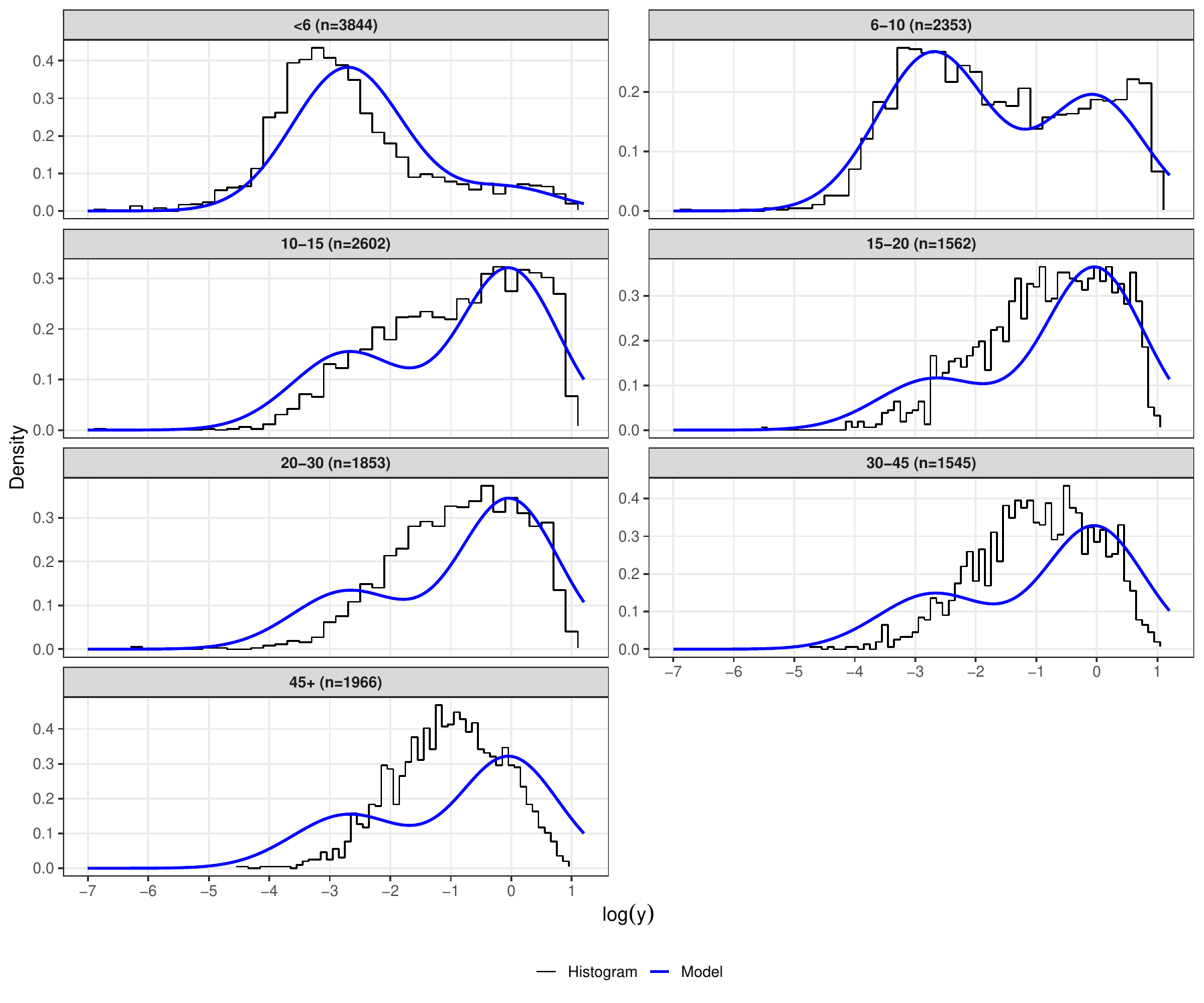}
    \caption{Apical membrance antigen 1 (AMA1): Histograms and fitted Gaussian mixture models for each age group shown in the panels. The means and variances of the two Gaussian components are estimated using data from the age group $<6$ years and are then held fixed for the remaining age groups, while the mixing probabilities are re-estimated.}
    \label{fig:compare_gmm_ama}
\end{figure}

\begin{figure}[H]
    \centering
    \includegraphics[width=0.8\linewidth]{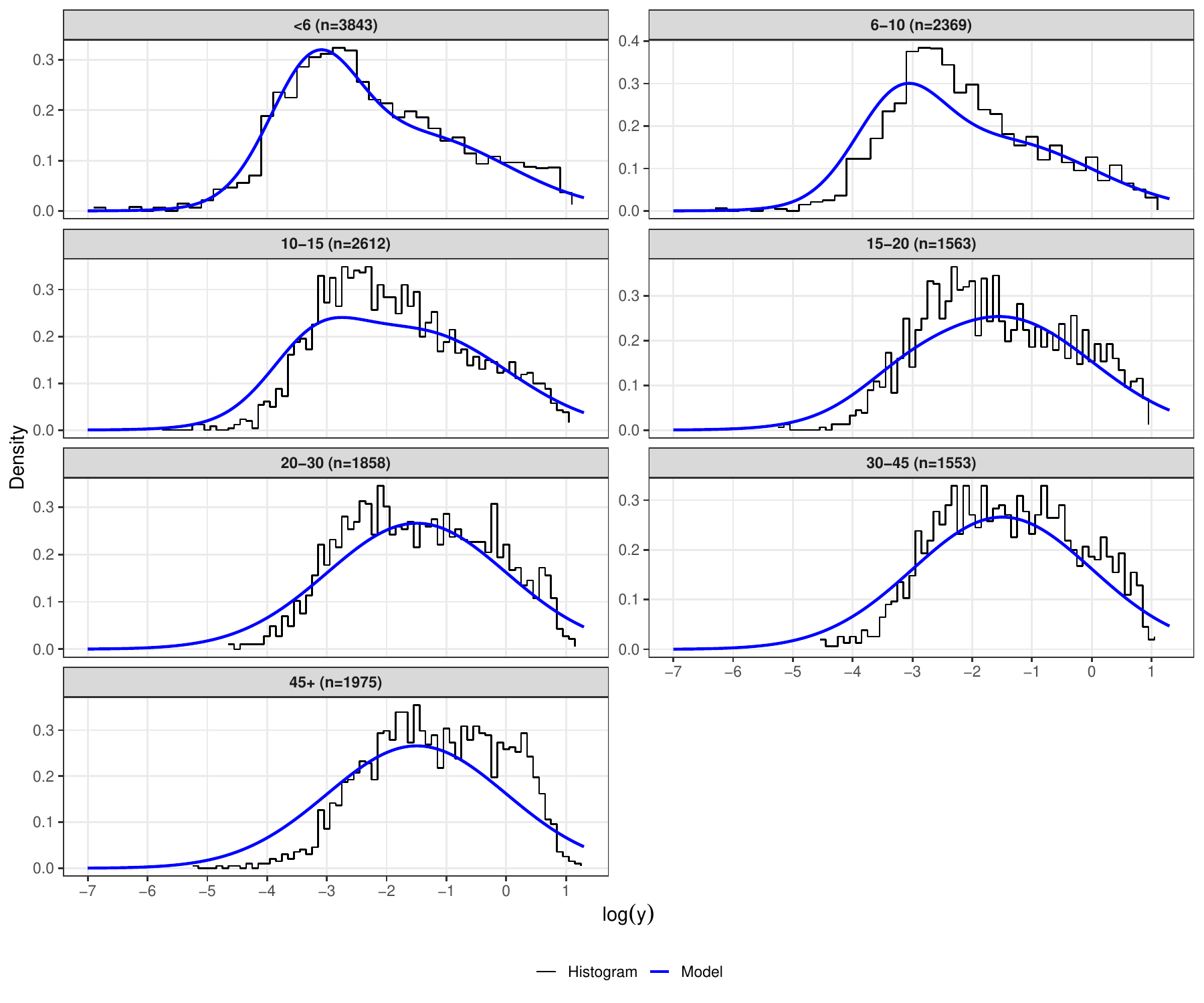}
    \caption{Merozoite surface protein 1 (MSP1): Histograms and fitted Gaussian mixture models for each age group shown in the panels. The means and variances of the two Gaussian components are estimated using data from the age group $<6$ years and are then held fixed for the remaining age groups, while the mixing probabilities are re-estimated.}
    \label{fig:compare_gmm_msp}
\end{figure}

\end{document}